\documentclass[preprint,hidelinks]{elsarticle}

\makeatletter
\def\ps@pprintTitle{%
 \let\@oddhead\@empty
 \let\@evenhead\@empty
 \def\@oddfoot{}%
 \let\@evenfoot\@oddfoot}
\makeatother
\usepackage{latexsym}
\usepackage{amssymb}
\usepackage{amsfonts}
\usepackage{amstext}
\usepackage{amsmath}
\usepackage{multicol}
\usepackage{xcolor}
\usepackage{array}
\usepackage{hyperref}
\usepackage{mathtools}
\usepackage{geometry}
\usepackage{bm}
\usepackage{tikz-cd}
\usepackage{pgfplots}
\usepackage{graphicx}
\usepackage{epstopdf}
\usepackage[hypcap=false]{caption}
\usepackage{subcaption}
\usepackage{pstricks}
\usepackage{booktabs}
\usepackage{natbib}
\usepackage{starfont}
\usepackage{amsthm}

\usepackage{calc}
\newlength\lone

\newcommand{\R}{\mathcal{R}}
\newcommand{\Q}{\mathcal{Q}}

\newcommand{\bx}{\mathbf{x}}
\newcommand{\ov}{\overline}

\newcommand{\bu}{\mathbf{u}}
\newcommand{\blangle}{\big\langle}
\newcommand{\Blangle}{\Big\langle}
\newcommand{\brangle}{\big\rangle}
\newcommand{\Brangle}{\Big\rangle}

\newcommand{\field}[1]{\mathbb{#1}}

\newcommand{\ra}[1]{\renewcommand{\arraystretch}{#1}}

\definecolor{greenx}{rgb}{0.4660,0.6740,0.1880}
\definecolor{purpdot}{rgb}{0.4940,0.1840,0.5560}

\theoremstyle{definition}
\newtheorem{crit}{Criterion}
\newtheorem{crit2}{Criterion}

\newtheorem{proposition}{Proposition}

\begin{document}

\begin{frontmatter}

\title{Heat transport in a hierarchy of reduced-order convection models}

\author{Matthew L. Olson\corref{cor1}}
\ead{mlolson@umich.edu}
\author{Charles R. Doering}
\ead{doering@umich.edu}

\cortext[cor1]{Corresponding Author.}

\begin{abstract}
Reduced-order models (ROMs) are systems of ordinary differential equations (ODEs) designed to approximate the dynamics of partial differential equations (PDEs). In this work, a distinguished hierarchy of ROMs is constructed for Rayleigh's 1916 model of natural thermal convection. These models are distinguished in the sense that they preserve energy and vorticity balances derived from the governing equations, and each is capable of modeling zonal flow. Various models from the hierarchy are analyzed to determine the maximal heat transport in a given model, measured by the dimensionless Nusselt number, for a given Rayleigh number. Lower bounds on the maximal heat transport are ascertained by computing the Nusselt number among equilibria of the chosen model using numerical continuation. A method known as sum-of-squares optimization is applied to construct upper bounds on the time-averaged Nusselt number. In this case, the sum-of-squares approach involves constructing a polynomial quantity whose global nonnegativity implies the upper bound along all solutions to a chosen ROM. The minimum such bound is determined through a type of convex optimization called semidefinite programming. For the ROMs studied in this work, the Nusselt number is maximized by equilibria whenever the Rayleigh number is sufficiently small. In this range of Rayleigh number, the equilibria maximizing heat transport are those that bifurcate first from the zero state. Analyzing this primary equilibrium branch provides a possible mechanism for the increase in heat transport near the onset of convection.
\end{abstract}


\end{frontmatter}

\section{Introduction} \label{sec:intro}

Thermal convection underlies a vast array of real-world phenomena, including atmospheric dynamics \cite{Stevens2005}, mantle convection \cite{Ogawa2008}, and stellar physics \cite{Miesch2005, Spiegel71a}. The fundamental mathematical model for this process is Rayleigh--B\'enard convection \cite{Rayleigh1916}, consisting of a horizontal layer of fluid heated from below. The temperature difference between the upper and lower boundaries drives heat transport across the layer. Convection occurs when the imposed temperature gradient is sufficiently large so that buoyancy overcomes viscous damping. 

While Rayleigh--B\'enard convection is a fundamental physical process, its complexity poses great challenges to obtaining exact solutions or performing comprehensive analysis. These challenges have drawn the attention of countless researchers in a variety of fields. The result is a rich literature of theoretical studies \cite{Malkus1954a, Spiegel1960, Busse69, Howard63, Whitehead2011}, numerical simulations \cite{Johnston2009, Scheel2014, Iyer2020}, and real-world experiments \cite{Niemela2000, Niemela2003, Urban2011}. A prominent focus of such studies is determining the dependence of the heat transport rate on the magnitude of the imposed temperature gradient. Casting the dimensional quantities in terms of dimensionless variables facilitates the search for this relationship by allowing the dynamics to be expressed in terms of a small number of characteristic parameters. The standard dimensionless measure of the temperature gradient is the Rayleigh number, Ra. One common measure of the heat transport rate is the dimensionless Nusselt number, Nu, given by the ratio of total heat transport across the fluid layer to that of pure conduction. Other dimensionless parameters include the Prandtl number, Pr, that quantifies the material properties of the fluid, and the domain aspect ratio. 

Determining the relationship between Nu and Ra presents a significant challenge, especially at large Rayleigh number. It is often presumed that in the large-Ra limit, the Nusselt number along any statistically steady flow is asymptotic to a power function of Ra, that is 
\begin{equation} \label{eq:NuRascale}
\mbox{Nu} \sim \mbox{Pr}^p \mbox{Ra}^q, \quad \mbox{Ra} \to \infty,
\end{equation}
for some scaling exponents $p$ and $q$. A major theoretical objective in the study of thermal convection is to determine the value of these exponents (if such a relationship exists). Numerous investigators have addressed this problem, and a brief history of such works is presented below. Malkus \cite{Malkus1954b} and Priestley \cite{Priestley1954} each proposed scaling arguments that suggest $p = 0$, $q = 1/3$ in \eqref{eq:NuRascale}. 
Dimensional arguments offered by Spiegel \cite{Spiegel1963, Spiegel71a} imply that instead $p = q = 1/2$. 

Numerous studies have been conducted to determine the correct scaling law at large Ra, including laboratory experiments \cite{Niemela2003,Urban2011} and direct numerical simulations \cite{Johnston2009,Stevens2011,Iyer2020}. Recent results suggest that the $q = 1/3$ scaling law persists over at least five orders of magnitude (from \linebreak $\mbox{Ra}\approx$ $10^{10}$ to Ra $\approx 10^{15}$) \cite{Iyer2020, Doering2020}. Whether a transition to the $q = 1/2$ regime occurs is yet to be rigorously determined. Both simulations and experiments introduce uncertainties that cannot be completely controlled. An alternative theoretical approach is to derive upper or lower bounds on $\mbox{Nu}(\mbox{Ra}, \mbox{Pr}, A)$ directly from the equations of motion. Variational methods are used in \cite{Wen2015, Whitehead2011}, where the maximal heat transport is rigorously bounded for two dimensional Rayleigh--B\'enard convection with stress-free isothermal boundaries. The tightest variational bounds establish that $q \leq 5/12$ in \eqref{eq:NuRascale} \cite{Whitehead2011}. This does not rule out the $q = 1/2$ scaling in the general case, but it implies that if this ``ultimate regime" of convection exists, it must either occur for domains with no-slip boundaries or include fluid motions that are fully three-dimensional.


The Galerkin-truncated models analyzed in this work approximate those of Rayleigh's PDE through a truncated Fourier series expansion. Prior truncated models of this form \cite{Lorenz63, Howard86, Thiffeault1995, Gluhovsky2002} have led to advancements in the understanding of chaos theory and mean-flow instabilities. Approximating the PDE by an ODE system allows upper bounds on the truncated version of the Nusselt number to be established through the application of recently developed techniques for bounding time-averaged quantities in polynomial dynamical systems \cite{Cherny2014}. In this approach, polynomial inequalities are constructed whose non-negativity implies an upper bound on the Nusselt number for all solutions to the ODE model. Introducing a carefully chosen polynomial \emph{auxiliary function} that depends on the state variables of the ODE allows meaningful global bounds to be obtained without solving the ODEs. Writing these inequalities as sum-of-squares constraints enables the feasible minimization over the class of polynomial auxiliary functions of fixed degree $d$. This approach has been employed in previous work \cite{Goluskin2018,Goluskin2019,Fantuzzi2016}, including a study of an eight-mode Galerkin-truncated model of Rayleigh--B\'enard convection that forms the foundation for this paper \cite{Olson2020}.

In this paper, a hierarchy of reduced-order models is developed to approximate the physics of Rayleigh's PDE \cite{Rayleigh1916} in a horizontally periodic domain with stress-free isothermal boundary conditions. Each model in this hierarchy satisfies a set of energy and vorticity balance laws derived from the PDE and includes horizontal shear modes capable of producing zonal flow. Such models build on the results of previous studies \cite{Lorenz63, Thiffeault1995, Olson2020} by forming successively closer approximations to the governing equations. Analyzing the heat transport properties of models in the chosen hierarchy supports the hypothesis that equilibria provide maximal heat transport for Rayleigh--B\'enard convection. The development of the equilibria near the onset of convection with increasing model size is explored to better understand the mechanism of heat transport by the equilibria.

The rest of the paper is structured as follows. The governing equations for Rayleigh--B\'enard convection are detailed in \S\ref{sec:RBC}. The general construction of Galerkin-truncated models approximating these governing equations is given in \S\ref{sec:ROM}, and the particular hierarchy of truncated models we consider is presented in \S\ref{sec:HK}. In \S\ref{sec:HKpart} we analyze the equilibria of these truncated models, providing candidate solutions for---and lower bounds on---the optimal heat transport. The general framework of our bounding approach is described in \S\ref{sec:SOS} and this technique is used to establish numerical upper bounds on heat transport in \S\ref{sec:UB}. Numerical and analytical details are provided in the appendices.

\section{Rayleigh--B\'enard convection} \label{sec:RBC}

The truncated models studied in this work are approximations to the PDEs governing Rayleigh's 1916 model of two-dimensional thermal convection \cite{Rayleigh1916}. To construct Rayleigh's model, consider a fluid in a horizontally periodic domain $(x,z) \in [\pi A d] \times [\pi d]$ with velocity $\mathbf{u}(\bx,t)$, pressure $p(\bx,t)$, density $\rho$, and temperature $T(\bx,t)$. Impermeable walls along the upper and lower boundaries of the domain are held at the fixed temperatures $T_t$ and $T_b$, respectively, with the resulting temperature drop $\Delta := T_b-T_t$. Density variations are assumed to be sufficiently small such that the dynamics of the fluid can be accurately modeled under the Boussinesq approximation. The pertinent assumptions are that the kinematic viscosity $\nu$, gravitational constant $g$, and thermal diffusivity $\kappa$ are constant, and that the density is fixed at $\rho_0$ except in the term representing the buoyancy force. In the buoyancy term, the density is assumed to follow the linear profile $\rho = \rho_0 (1 - \alpha (T - T_b))$, where $\alpha$ is the coefficient of thermal expansion. The relevant material parameters form two dimensionless groups, typically represented by the Prandtl number, $\sigma$, and the Rayleigh number, Ra:
\begin{equation}
\sigma := \frac{\nu}{\kappa}, \qquad \text{Ra} := \frac{g \alpha (\pi d)^3 \Delta}{\nu \kappa}.
\end{equation}
The Navier--Stokes equations, nondimensionalized using length scale $d$, time scale $d^2/\kappa$, and temperature scale $\Delta$, are expressed as \cite{Chand1961}
\begin{align} \label{eq:NSE}
\partial_t \bu + \bu \cdot \nabla \bu &= - \nabla p + \sigma \nabla^2 \bu + \sigma \mathcal{R} \pi T \hat{\mathbf{z}}, \\
\nabla \cdot \bu = 0, \\
\partial_t T + \bu \cdot \nabla T &= \nabla^2 T,
\end{align}
where $\mathcal{R} := \text{Ra}/\pi^4$ is the modified Rayleigh number. 

We consider 2D Rayleigh--B\'enard convection because this allows the construction of a stream function $\psi$ that is related to the velocity variables by $(u,w) = (\partial_z \psi, -\partial_x \psi)$. The temperature function is shifted so that the fixed boundary temperatures are $T = 1$ along the bottom and $T = 0$ along the top. After this transformation, the linear temperature profile in a quiescent fluid becomes $T_c := 1 - z/\pi$. Then let the dimensionless temperature deviation function $\theta$ be defined as $\theta := \pi (T_c - T)$. The dimensionless Boussinesq equations, expressed in terms of $\psi$ and $\theta$, are given by
\begin{align} 
\partial_t \nabla^2 \psi - \{\psi, \nabla^2 \psi\} &= \sigma \nabla^4 \psi + \sigma \R \partial_x \theta, \label{eq:BE_psi}\\
\partial_t \theta - \{\psi, \theta\} &= \nabla^2 \theta + \partial_x \psi, \label{eq:BE_theta}
\end{align}
where $\{f,g\}:= \partial_x f \partial_z g - \partial_z f \partial_x g$ denotes the Jacobian of the functions $f$ and $g$. In other ROM convection studies, $\theta$ is sometimes rescaled by a factor of $\R$, thereby moving the Rayleigh number to the $\partial_x \psi$ term in \eqref{eq:BE_theta}. However, this rescaling does not appear to provide any benefit from a numerical conditioning perspective.

At the upper and lower boundaries, we impose stress-free boundary conditions. Such boundary conditions have been employed in other studies of truncated models of Rayleigh--B\'enard convection \cite{Howard86,Thiffeault1995, Gluhovsky2002}, in part due to the convenient Fourier expansion of $\psi$. Stress-free boundary conditions require that
\begin{equation} \label{eq:stressFree}
\psi = \partial_{zz} \psi = 0, \quad z = 0, \pi.
\end{equation}
Fixed-temperature boundary conditions imply that $\theta$ obeys Dirichlet boundary conditions:
\begin{equation}
\theta = 0, \quad z = 0, \pi.
\end{equation}
In addition, all variables are periodic in $x$ with period $A \pi$. 

The primary emergent quantity of interest is the Nusselt number, Nu, defined as the ratio of total heat transport to conductive heat transport, averaged over the domain and over infinite time. The Nusselt number can be written in terms of $\psi$ and $\theta$ by averaging this ratio over the fluid domain, resulting in the expression
\begin{equation} \label{eq:Nusselt}
\text{Nu} = 1 + \blangle \ov{\theta \partial_x \psi} \brangle,
\end{equation}
where the spatial and temporal averages are given by
\begin{equation} \label{eq:spaceAvg}
\langle f \rangle := \frac{1}{A \pi^2} \int_0^\pi \int_0^{A\pi} f(x,z) \, {\rm d}x \, {\rm d}z, 
\end{equation}
and
\begin{equation} \label{eq:timeAvg}
\ov{f} := \lim_{\tau \to \infty} \frac{1}{\tau} \int_0^\tau f(t)\, {\rm d}t,
\end{equation}
assuming the limit exists. An equivalent expression for the Nusselt number is obtained by averaging over a horizontal slice, yielding
\begin{equation} \label{eq:Nusselt2}
\text{Nu} = 1 + \left[ \ov{\partial_z \langle \theta \rangle_x}(z) + \ov{\langle \theta \partial_x \psi \rangle_x}(z) \right],
\end{equation}
where the horizontal average is given by
\begin{equation} \label{eq:horizAvg}
\langle f \rangle_x := \frac{1}{A \pi} \int_0^{A\pi} f(x) \, {\rm d}x.
\end{equation}
The quantities \eqref{eq:Nusselt} and \eqref{eq:Nusselt2} are equivalent along statistically steady flows of the Boussinesq equations. This property will be used in the next section as a criterion for determined the quality of truncated models.

\section{Truncated model construction} \label{sec:ROM} 

A substantial body of research has been devoted to studying the dependence of the heat transport on the Rayleigh number in Rayleigh--B\'enard convection. Even so, there remains a gap between the rigorous upper bounds on the Nusselt number derived from the equations of motion, and the maximal heat transport obtained from laboratory experiments and numerical simulations. As a complementary approach, one may construct \emph{reduced-order models}, or ROMs. These are finite systems of ordinary differential equations (ODEs) derived from the governing equations. Ideally, ROMs should approximate the dynamics of fluid convection. 

Various ODE models have been derived for Rayleigh--B\'enard convection, beginning with the atmospheric model of Saltzman~\cite{Saltzman62} that inspired the seminal study by Lorenz of a three-dimensional ROM now known as the Lorenz equations~\cite{Lorenz63}. Although the Lorenz equations are a simplified model of Rayleigh's PDE, the ODE system correctly predicts the minimal Rayleigh number where convection can occur and accurately models the physics of 2D Rayleigh--B\'enard convection near the onset of convection. Following the work of Lorenz, many other ROMs have been studied as simplified convection models~\cite{Howard86,Thiffeault1995,Thiffeault1996,Hermiz1995,Gluhovsky2002}. 

This section provides general details on the construction of reduced-order models for 2D Rayleigh--B\'enard convection with stress-free isothermal boundaries in a horizontally periodic domain. Stress-free boundaries are chosen for models constructed in this work in part because this allows expansion in terms of the Fourier basis. Readers interested in the hierarchy of models introduced in this paper but not in the general model construction details may wish to skip to section \S\ref{sec:HK}.

\subsection{Galerkin expansions and reduced-order models} \label{sec:Galerkin}

The derivation of the Lorenz equations is an example of a general technique known as \emph{Galerkin expansion}, where dependent variables are expanded in terms of an orthogonal set of basis functions that each satisfy the boundary conditions, producing an ROM. Square-integrable functions $\psi$ and $\theta$ satisfying the boundary conditions are given by the series
\begin{equation}
\begin{aligned} \label{eq:Galerkin}
\psi(x,z,t) &= \sum_{m=0}^\infty \sum_{n=1}^\infty \left[ a_{mn}(t) \cos mkx +b_{mn}(t) \sin mkx \right] \sin nz, \\
\theta(x,z,t) &= \sum_{m=0}^\infty \sum_{n=1}^\infty \left[ c_{mn}(t) \cos mkx  + d_{mn}(t) \sin mkx \right] \sin nz.
\end{aligned}
\end{equation}
Here $k = 2/A$ is the fundamental horizontal wavenumber for a domain of aspect ratio $A$. The subscripts on the coefficients in the above expansions correspond to the indices on the horizontal and vertical mode numbers of the associated Fourier modes. Inserting the expansions~\eqref{eq:Galerkin} into the Boussinesq equations~\eqref{eq:BE_psi}--\eqref{eq:BE_theta} and projecting the resulting expression onto each basis element results in a system of ODEs describing the time evolution of the Fourier mode amplitudes. These amplitude functions will be referred to as ``modes" whenever the context is clear. 

Practical applications require that the series expansions be truncated to some finite number of terms in \eqref{eq:Galerkin}, yielding a finite system of ODEs that approximates the full PDE. Nonlinear interactions between the Fourier modes yield terms outside the span of the modes in the truncated model; these excess terms are discarded when projecting onto only the included modes.

\subsection{Model simplification} \label{sec:simplifyModel}

The series expansions \eqref{eq:Galerkin} for $\psi$ and $\theta$ include two terms for each horizontal wavenumber, differing from each other only in their horizontal phase. For models in this work we consider solutions of fixed horizontal phase by setting $a_{11} = 0$ in \eqref{eq:Galerkin} as in the derivation of the Lorenz equations. This allows the inclusion of a larger spectrum of horizontal wavenumbers in the resulting system given a particular model size. Making the above choice determines the horizontal phase of all terms in~\eqref{eq:Galerkin} to maintain consistency upon substitution into the governing equations. For instance, the $b_{11}$ ODE contains $c_{11}$, but not $d_{11}$, so for consistency this implies $d_{11} = 0$. Continuing this process by selecting the consistent phase in each term of \eqref{eq:Galerkin} results in:
\begin{equation}
\begin{aligned}
a_{mn} &\equiv d_{mn} \equiv 0, \mbox{ for } m + n \mbox{ even},\\
b_{mn} &\equiv c_{mn} \equiv 0, \mbox{ for } m + n \mbox { odd}.
\end{aligned}
\end{equation}
We henceforth discard all terms with inconsistent with the chosen horizontal phase condition, and let $\psi_{mn}$ and $\theta_{mn}$ be the nonzero coefficients remaining in the Fourier expansions. For example, $\psi_{11} = b_{11}$, $\psi_{12} = a_{12}$, and so on.
The truncated models studied here take the form
\begin{equation}
\begin{aligned} \label{eq:Galerkin2}
\psi(x,z,t) = \sum_{\mathclap{(m,n) \in S_\psi}} \psi_{mn}(t) \, f_\psi(mkx) \sin nz, \\
\theta(x,z,t) = \sum_{\mathclap{(m,n) \in S_\theta}} \theta_{mn}(t) \, f_\theta(mkx) \sin nz,
\end{aligned}
\end{equation}
where $f_\psi(x) = \sin x$ for each term such that the total wavenumber $m + n$ is even and $\cos x$ whenever $m + n$ is odd, due to the phase convention described above. Similarly, $f_\theta(x) = \cos(x)$ if $m + n$ is even and $\sin(x)$ if $m + n$ is odd. 

Fixing the horizontal phase restricts the space of solutions to the governing equations \eqref{eq:BE_psi}--\eqref{eq:BE_theta} to those exhibiting a particular symmetry about the points $(\pi/2k, \pi/2)$ and $(3\pi/2k, \pi/2)$. These points are the midpoints horizontally and vertically in their respective half-domain and lie at the center of the stable rolls that arise at the onset of convection. In each half-domain, the truncated $\theta$ expansion is symmetric along any line passing through the center point, and $\psi$ is antisymmetric about the same point. A corollary of this result is that all modes of the form $\psi_{0n}$ (also called shear modes) must have odd vertical wavenumber, while modes of the form $\theta_{0n}$ must have even vertical wavenumber.

ROMs constructed from \eqref{eq:Galerkin2} are capable of capturing the dynamics of \emph{zonal flow} if at least one mode of the form $\psi_{0n}$ is included in the truncated. Zonal flow occurs when mean horizontal flows near the top and bottom boundaries vertically shear the fluid~\cite{Goluskin2013}. 
This phenomenon has been observed in experiments of turbulent convection~\cite{Krishnamurti1981} and occurs in toroidal plasmas~\cite{Diamond2005} and planetary atmospheres~\cite{Busse1994}. Howard and Krishnamurti~\cite{Howard86} designed a six-ODE truncated model that is constructed by augmenting the Lorenz equations with the $\psi_{01}$, $\psi_{12}$ and $\theta_{12}$ modes, resulting in a truncated model that exhibits zonal flow. Yet Howard and Krishnamurti themselves observed nonphysical behavior in their model, including unbounded trajectories, that make it unsuitable for drawing any analogy with the heat transport of the PDE.
Thiffeault and Horton~\cite{Thiffeault1995, Thiffeault1996} suggested the inclusion of the $\theta_{04}$ mode that ensures conservation of mechanical energy in the dissipationless limit $\nu, \kappa \to 0$. Models exhibiting this property have bounded trajectories, and the truncated versions of the two expressions for the Nusselt number, \eqref{eq:Nusselt} and \eqref{eq:Nusselt2}, are equivalent along all trajectories \cite{Thiffeault1995}. The result is physically reasonable heat transport near the onset of convection. A similar adjustment to Howard and Krishnamurti's model was proposed by Hermiz \emph{et al.}~\cite{Hermiz1995} who included the $\psi_{03}$ mode so that solutions conserve the truncated version of total vorticity in the dissipationless limit. These ideas culminated in an eight-dimensional model introduced by Gluhovsky \emph{et al.} \cite{Gluhovsky2002} by adding both the $\theta_{04}$ and the $\psi_{03}$ mode to the Fourier expansions used to construct Howard and Krishnamurti's original model. This system is called the HK8 model because it is the minimal extension of Howard and Krishnamurti's model that restores these basic integral identities of the PDE \cite{Olson2020}.

\subsection{Model construction} \label{sec:constructModel}

We now describe the general form of reduced-order models of Rayleigh--B\'enard convection satisfying the above properties. Equations for such truncated models were previously given in \cite{Treve1982,Thiffeault1995}. Here we present a different form of these equations designed for programmatic construction. The code used to construct models in this work from the below equations can be found on GitHub\footnote{GitHub repository: \url{https://github.com/PeriodicROM/construct_roms}}.

Suppose a reduced-order model is defined by selecting a finite number of terms from the expansions of $\psi$ and $\theta$. Let $S_\psi$ and $S_\theta$ be the sets of all selected modal pairs $(m,n)$ that from the $\psi$ and $\theta$ series, respectively. 
Galerkin expansion yields the following ODEs for any given pair $(m,n)$ \cite{Thiffeault1995}:
\begin{align} 
    \dot{\psi}_{mn} &= -\sigma \rho_{mn} \psi_{mn} + (-1)^{m+n} (\sigma \R) \frac{mk}{\rho_{mn}} \theta_{mn} + Q_{mn}^{\psi}, \label{eq:psiTrunc} \\
     \dot{\theta}_{mn} &= -\rho_{mn} \theta_{mn} + (-1)^{m+n} (mk) \psi_{mn} + Q_{mn}^{\theta}, \label{eq:thetaTrunc}
 \end{align}
where $\rho_{mn} := (mk)^2 + n^2$ are the eigenvalues of $-\nabla^2$ in the Rayleigh--B\'enard domain, and $Q_{mn}^{\psi}$, $Q_{mn}^{\theta}$ consist of the sum of all quadratic terms in the corresponding ODE. The quadratic terms arise from the nonlinear terms of \eqref{eq:BE_psi}--\eqref{eq:BE_theta}. 
For fixed $(m,n)$, the terms in $Q_{mn}^{\psi}$ are proportional to $\psi_{pq}\, \psi_{rs}$ whose modal pairs lie in the set 
\begin{equation}
P_\psi[(m,n)] = \{((p,q), (r,s)) \in S_\psi \times S_\psi : m = |p \pm r|, n = |q \pm s|, (p,q) > (r,s) \}, \label{eq:psiPairs}
\end{equation}
where $(p,q) > (r,s)$ is the lexicographical ordering, defined by
\begin{equation} \label{eq:lexico}
(p,q) > (r,s) \iff p > r \mbox{ or } (p = r \mbox{ and } q > s).
\end{equation} 
The ordering restriction on $P_\psi[(m,n)]$ ensures that terms in $Q_{mn}^\psi$ are not double counted by commuting the modes. Modal pairs $ (\beta, \gamma) \in P_\psi$ can be combined with $\alpha := (m,n)$ to produce a ``compatible triplet" of modes $(\alpha, \beta, \gamma)$; in this triplet, two of the horizontal wavenumbers must sum to the other, and likewise for the vertical wavenumber. All nonlinear terms take this form because this is the condition for two Fourier modes to produce another after multiplication. Additionally, compatible triples admit the following symmetry property for $\alpha > \beta > \gamma$:
\begin{equation} 
(\beta, \gamma) \in P_\psi[\alpha] \implies (\alpha, \gamma) \in P_\psi[\beta] \mbox{ and } (\alpha, \beta) \in P_\psi[\gamma].
\end{equation}

The quadratic terms in the $\psi_{mn}$ equations are then given by
\begin{equation}
    Q_{mn}^\psi = \frac{k}{\rho_{mn}} \sum_{P_\psi[(m,n)]} \frac{\mu_1}{d} \big[B_{pmr} B_{snq} (p s) - B_{qns} (q r) \big] (\rho_{p q} - \rho_{r s}) \psi_{p q} \psi_{r s},  \label{eq:psiTruncQuad}
\end{equation}
where $B, \mu_1$ and $d$ are defined by
\begin{align} \label{eq:Bmatrix}
    B_{ijk} &= \begin{cases}
        -1, & i = j + k, \\
        1, & \mbox{ else},
        \end{cases} \displaybreak[1] \\
    \mu_1 &= \begin{cases}
        B_{pmr}, & (m+n) \ \mbox{ even}, (r+s) \  \mbox{ odd}, \\
        -B_{pmr},  & (m+n) \ \mbox{ odd},\; (r+s) \  \mbox{ odd}, \\
        -1, & {\rm else},
        \end{cases} \displaybreak[1] \\
    d &= \begin{cases}
        2, & p = 0 \ \mbox{or} \  r = 0, \\
        4, & \mbox{else}.
        \end{cases}
\end{align}
Similarly, the terms in $Q_{mn}^\theta$ are proportional to $\psi_{pq} \, \theta_{rs}$, represented by the set 
\begin{equation}
P_\theta[(m,n)] = \{((p,q), (r,s)) \in S_\psi \times S_\theta: m = |p \pm r|, n = |q \pm s|\}. \label{eq:thetaPairs}
\end{equation}
No ordering is needed on the pairs in $P_\theta$ since commuting the modal pairs yields a distinct term in the sum. Again, a compatible triplet can be formed of modes whose wavenumbers combine in the appropriate way. In this case $(\beta, \gamma) \in P_\theta[\alpha]$ implies $(\beta, \alpha) \in P_\theta[\gamma]$ (a different symmetry than that of $P_\psi$). The quadratic terms are then expressed as
\begin{equation} 
    Q_{mn}^\theta = k \sum_{P_\theta[(m,n)]} \frac{\mu_2}{d} \big[B_{pmr} B_{snq} (p s) - \mu_3 B_{qns} B_{rpm}(q r) \big] \psi_{pq} \theta_{rs}, \label{eq:thetaTruncQuad}
\end{equation}
where $B, \mu_1$ and $d$ are defined as above, and $\mu_2, \mu_3$ are defined by
\begin{align}
    \mu_2 &= \begin{cases}
    	\mu_3 \,B_{rpm}, & (m+n) \ \mbox{ even}, (r+s) \ \mbox{ odd}, \\
    	-B_{rpm} B_{pmr}, & (m+n) \ \mbox{ odd}, \;(r+s) \ \mbox{ even}, \\
    	B_{pmr}, & (m+n) \ \mbox{ even}, (r+s) \ \mbox{ even},\\
    	1 & \mbox{ else},
    	\end{cases}\displaybreak[1] \\
    \mu_3  &= \begin{cases}
        -1, & m=0, \\
        1, & \mbox{else},
        \end{cases} \displaybreak[1] \\
\end{align}
Writing the equations in the form \eqref{eq:psiTrunc}--\eqref{eq:thetaTrunc} allows convenient algorithmic construction of the truncated system. In this work, all ODE models were generated via the Python package \texttt{construct\_roms} developed for this study that outputs the model equations ODE models given only the sets $S_\psi$ and $S_\theta$.
 
Given an ROM of the form \eqref{eq:psiTrunc}--\eqref{eq:thetaTrunc}, a version of the Nusselt number can be defined by projecting either version of Nu described in \S\ref{sec:RBC} onto $S_\psi \cup S_\theta$. 
Since these measures of heat transport are approximations of the Nusselt number of the PDE, we use $N$ to denote the truncated version of the quantity Nu in each reduced model. When the truncated Fourier series are inserted into the volume-averaged definition of Nu~\eqref{eq:Nusselt}, orthogonality reduces the expression to
\begin{equation} \label{eq:N_gen}
N = 1 + \tfrac{1}{4} \sum_{\mathclap{(m,n) \in S_\psi \cap S_\theta}} (-1)^{m+n} (mk)\, \ov{\psi_{mn} \theta_{mn}}.
\end{equation}
Alternatively, deriving the expression for $N$ from~\eqref{eq:Nusselt2} yields
\begin{equation} \label{eq:N2_gen}
N = 1 + \sum_{\mathclap{(0,2n) \in S_\theta}} (2n) \, \ov{\theta_{0,2n}}.
\end{equation}
Whether the two definitions of $N$ are equivalent in the long-time average depends on the choice of modes in the truncation model. Equivalence of these expressions is a desirable property for ROMs because an analogous result holds for the expressions \eqref{eq:Nusselt} and \eqref{eq:Nusselt2} along statistically steady solutions to the Boussinesq equations. Motivated by the study of optimal heat transport in Rayleigh--B\'enard convection, we seek to determine the maximum Nusselt number, $N^*$, given by
\begin{equation} \label{eq:Nmax}
N^* := \sup_{\bx(t)} N,
\end{equation}
where the supremum is taken over all solutions $\bx(t)$ for a given reduced-order model, and generally depends on the parameters $\R, \sigma$ and $k$.

\section{The HK hierarchy} \label{sec:HK}

Constructing a Galerkin-truncated model of the Boussinesq equations requires selecting a finite set of modes in the Fourier expansions for $\psi$ and $\theta$. There is no universally accepted way to choose which modes to include, although a few guidelines have been established to promote consistency with the Boussinesq equations. Authors in previous studies of low-order models~\cite{Treve1982,Thiffeault1995,Hermiz1995,Gluhovsky2002} have suggested imposing criteria based on conservation laws derived from the Boussinesq equations. In this section, we examine a hierarchy of reduced-order convection models obeying such properties. Solutions to these models are then analyzed in \S\ref{sec:HKpart} and upper bounds on heat transport are computed in \S\ref{sec:UB}.


The truncated models considered here are constructed according to \eqref{eq:psiTrunc}--\eqref{eq:thetaTrunc} with the sets of modal pairs $(m,n)$ of all selected modes for $\psi$ and $\theta$ given by $S_\psi$ and $S_\theta$. We consider models that are distinguished in the sense that each model satisfies the following  energy, temperature, and vorticity balance laws:
\begin{align}
\partial_t \left[ \tfrac{1}{2} \left\langle \lvert\nabla \psi \rvert^2 \right\rangle + \sigma \R \langle z \theta \rangle \right] &=  \, \sigma \R \left\langle z \nabla^2 \theta \right\rangle - \sigma \left\langle (\nabla^2 \psi )^2 \right\rangle, \label{eq:EConsD_1} \\[6pt]
\partial_t \langle \theta \rangle &=  \, \left\langle \nabla^2 \theta \right\rangle, \label{eq:TConsD_1}\\[6pt]
\partial_t \left\langle \nabla^2 \psi \right\rangle &= \, \sigma \left\langle \nabla^4 \psi \right\rangle.\label{eq:VConsD_1}
\end{align}
An ROM satisfies a given balance law if equality holds after substituting the truncated Fourier series into the balance equation. As proved in \ref{sec:ConsLaws}, a truncated model satisfies \eqref{eq:EConsD_1}--\eqref{eq:VConsD_1} if modes are selected according to the following criteria:
\begin{crit2}[Energy balance]\label{rule:Econs}
If $(m,n) \in S_\psi \cap S_\theta,$ then $(0,2n) \in S_\theta$ \cite{Thiffeault1995}.
\end{crit2}
\begin{crit2}[Vorticity balance]\label{rule:Vcons}
If $(p,q) \in S_\psi$ and $(p,s) \in S_\psi$, then $(0,\lvert q-s \rvert) \in S_\psi$ if and only if ${(0,q+s) \in S_\psi}$.
\end{crit2}
\noindent The benefits of considering such models include boundedness of all trajectories of the ODEs \cite{Thiffeault1995} and equivalence of the two definitions of the time-averaged Nusselt number, \eqref{eq:N_gen} and \eqref{eq:N2_gen}.

We direct our focus to the subset of truncated models including one or more ``shear modes"---i.e., those of the form $\psi_{0n}$. The smallest model obeying each of the above criteria is the HK4 model, a modified version of the Lorenz equations that includes the Fourier modes $\psi_{01}$, $\psi_{11}$, $\theta_{02}$, and $\theta_{11}$. In this model, the shear mode decays exponentially along all orbits since it satisfies the simple uncoupled equation $\dot{\psi}_{01} = -\sigma \psi_{01}$, so that the fully developed dynamics are indistinguishable from those of the Lorenz equations. The next smallest model---the first to exhibit nontrivial shear flow---is the HK8 model studied in \cite{Olson2020}. We construct a hierarchy of distinguished models with shear that build on these two initial cases. First, define an ordering on the mode pairs as $(m_1,n_1) > (m_2,n_2)$ if and only if
\begin{equation} \label{eq:modeOrder}
\begin{aligned}
m_1+n_1 > m_2+n_2 \mbox{ or } (m_1+n_1 = m_2+n_2 \mbox{ and } n_1 > n_2).
\end{aligned}
\end{equation}
Let HK$M_i$ be the $i^{th}$ model in the hierarchy, containing $M_i$ modes. To construct the next model:
\begin{enumerate}
\item Find the smallest pair with $m,n > 0$ not included in HK$M_i$, according to \eqref{eq:modeOrder}.
\item Add the corresponding modes $\psi_{mn}$ and $\theta_{mn}$ to the system.
\item Add $\theta_{0,2n}$ and $\psi_{0,2n-1}$, if not already included in HK$M_i$.
\end{enumerate}
The final condition ensures that all models in the hierarchy are distinguished according to the above \hyperref[rule:Econs]{energy} and \hyperref[rule:Vcons]{vorticity} rules. In the limit as $M_i \to \infty$, the procedure enumerates all pairs $(m,n)$ with strictly positive indices, all shear modes with odd $n$, and all temperature modes with $m = 0$ and even $n$. Therefore all Fourier modes that satisfy the phase convention are included for sufficiently large $M_i$. The modes in the truncated Fourier series for several models in the hierarchy are listed in Table \ref{tab:HKModes}, and a schematic of the selection procedure is shown in Figure \ref{fig:modePairs}.

\begin{table}[!ht]
\centering
\ra{1.1}
\caption{Additional modes required to construct each HK model from the previous one in the hierarchy up to $M_i = 44$. Each model in the table includes all modes listed in the lines preceding it.} \label{tab:HKModes}
\begin{tabular}{ccccc}
\toprule
Model & Additional modes && Model & Additional Modes \\
\cmidrule{1-2} \cmidrule{4-5}
HK4 & $\psi_{01}, \psi_{11}, \theta_{02}, \theta_{11}$ &&  HK26 & $\psi_{32}, \theta_{32}$ \\
HK8 & $\psi_{03}, \psi_{12}, \theta_{04}, \theta_{12}$ && HK28 & $\psi_{41}, \theta_{41}$ \\
HK10 & $\psi_{21}, \theta_{21}$ && HK32 & $\psi_{09}, \psi_{15}, \theta_{0,10}, \theta_{15}$ \\
HK14 & $\psi_{05}, \psi_{13}, \theta_{06}, \theta_{13}$ && HK34 & $\psi_{24}, \theta_{24}$ \\
HK16 & $\psi_{22}, \theta_{22}$ && HK36 & $\psi_{33}, \theta_{33}$  \\
HK18 & $\psi_{31}, \theta_{31}$ && HK38 & $\psi_{42}, \theta_{42}$ \\
HK22 & $\psi_{07}, \psi_{14}, \theta_{08}, \theta_{14}$ && HK40 & $\psi_{51}, \theta_{51}$ \\
HK24 & $\psi_{23}, \theta_{23}$ && HK44 & $\psi_{0,11}, \psi_{16}, \theta_{0,12}, \theta_{16}$\\
\bottomrule
\end{tabular}
\end{table}

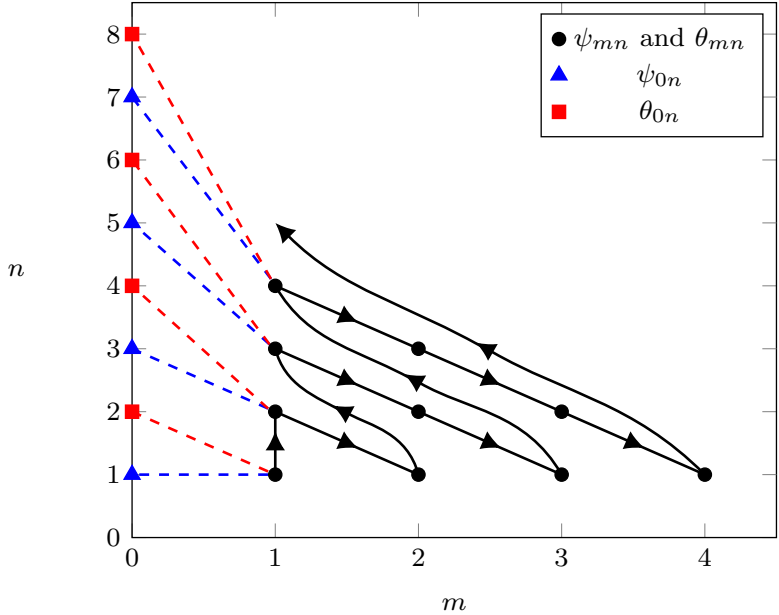
\begin{figure}[!ht]
\centering
\begin{tikzpicture}[>={Latex[scale=1,length=6,width=6]},scale=1.25,font=\footnotesize]
\begin{axis}[nodes near coords,nodes near coords align={below right},
			xmin=0, xmax=4.5, ymin=0, ymax=8.5,
			xlabel=$m$, ylabel=$n$, ylabel style={rotate=-90},
			extra y ticks= {1,3,5,7},
			scatter/classes={a={mark=*,black},
				b={mark=triangle*,blue,scale=1.4},
				c={mark=square*,red,scale=1}},
			legend entries={$\psi_{mn}$ and $\theta_{mn}$,$\psi_{0n}$,$\theta_{0n}$}]
				\addlegendimage{only marks,mark=*, black} 	
				\addlegendimage{only marks,mark=triangle*,blue, every mark/.append style={scale=1.4}}
				\addlegendimage{only marks, mark=square*,red,every mark/.append style={scale=1}}
			after end axis/.code={
				\draw[->,thick] (axis cs:1,1) -- (axis cs:1,1.7); 
					\draw[thick] (axis cs:1,1.5) -- (axis cs:1,2);
				\draw[-,dashed,blue,thick] (axis cs:.95,1) -- (axis cs:0,1); 
				\draw[-,dashed,red,thick] (axis cs:.95,1.05) -- (axis cs:0,2); 
			    \draw[->,thick] (axis cs:1,2) -- (axis cs:1.58,1.42); 
			    	\draw[thick] (axis cs:1.5,1.5) -- (axis cs:2,1);
		    	\draw[-,dashed,blue,thick] (axis cs:.95,2.05) -- (axis cs:0,3); 
			    \draw[-,dashed,red,thick] (axis cs:.95,2.05) -- (axis cs:0,4); 
			    \draw[thick] (axis cs:2,1) edge[out=105,in=-28,->]  (axis cs: 1.4,2.1); 
   			    	\draw[thick] (axis cs:1.5,1.97) edge[out=152,in=-75,-]  (axis cs: 1,3);
   			   \draw[-,dashed,blue,thick] (axis cs:.96,3.07) -- (axis cs:0,5); 
   			    \draw[-,dashed,red,thick] (axis cs:.96,3.07) -- (axis cs:0,6); 
			    \draw[->,thick] (axis cs:1,3) -- (axis cs:1.58,2.42); 
   			    	\draw[thick] (axis cs:1.5,2.5) -- (axis cs:2,2);
   			    \draw[->,thick] (axis cs:2,2) -- (axis cs:2.58,1.42); 
  			    	\draw[thick] (axis cs:2.5,1.5) -- (axis cs:3,1);
  			    \draw[thick] (axis cs:3,1) edge[out=120,in=-30,->]  (axis cs: 1.9,2.6); 
   			    	\draw[thick] (axis cs:2,2.46) edge[out=150,in=-60,-]  (axis cs: 1,4);
   			   	\draw[-,dashed,blue,thick] (axis cs:.97,4.11) -- (axis cs:0,7); 
   			    \draw[-,dashed,red,thick] (axis cs:.97,4.11) -- (axis cs:0,8); 
  			    \draw[->,thick] (axis cs:1,4) -- (axis cs:1.58,3.42);
 			    	\draw[thick] (axis cs:1.5,3.5) -- (axis cs:2,3);
 			    \draw[->,thick] (axis cs:2,3) -- (axis cs:2.58,2.42);
			    	\draw[thick] (axis cs:2.5,2.5) -- (axis cs:3,2);
			    \draw[->,thick] (axis cs:3,2) -- (axis cs:3.58,1.42);
   			    	\draw[thick] (axis cs:3.5,1.5) -- (axis cs:4,1);
   			    \draw[thick] (axis cs:4,1) edge[out=135,in=-30,->]  (axis cs: 2.4,3.1); 
   			       	\draw[thick] (axis cs:2.5,2.96) edge[out=150,in=-45,->]  (axis cs: 1,5);
				}]
\addplot [scatter,only marks,scatter src=explicit symbolic]
	coordinates {
		(1,1) [a]
		(0,1) [b]
		(0,2) [c]
		(1,2) [a]
		(2,1) [a]
		(0,3) [b]
		(0,4) [c]
		(1,3) [a]
		(2,2) [a]
		(3,1) [a]
		(0,5) [b]
		(0,6) [c]
		(1,4) [a]
		(2,3) [a]
		(3,2) [a]
		(4,1) [a]
		(0,7) [b]
		(0,8) [c]};
\end{axis}
\end{tikzpicture}
\caption{Schematic of the mode selection procedure for the first 10 models in the HK hierarchy, where arrows indicate the order of selection and modes connected with dashed lines are added simultaneously. Modes with $m = 0$ are added when the next $(m+n)$ shell is reached. Each model contains the modes of all previous models, beginning with the HK4 system that is represented by the point at $(1,1)$ along with the modes connected by dashed lines.} \label{fig:modePairs}
\end{figure}

\section{Particular solutions of models in the HK hierarchy} \label{sec:HKpart}

In this section, we examine particular solutions of reduced-order models in the HK hierarchy that was defined in the previous section. Solutions to the HK models provide candidates for the optimal heat transport in a given model and therefore bound $N^*$ from below. In this paper, we primarily focus on the equilibria of these models. Equilibria are of particular interest because it is conjectured that steady solutions maximize heat transport in the Boussinesq equations \cite{Wen2020}. Indeed, for models studied in this work, equilibria provide greater heat transport than any computed time dependent solutions for $\R$ well beyond the each model's capability to closely approximate the PDE.

\subsection{Equilibria of reduced-order models} \label{sec:HKEquil}

The types of equilibria present depend on the set of included modes. For all HK models, the zero equilibrium is globally attracting for sufficiently small $\R$. As the Rayleigh number increases, the zero state undergoes a series of pitchfork bifurcations. Each pitchfork bifurcation gives rise to a pair of equilibria via an instability in the variables $\psi_{mn}$ and $\theta_{mn}$, and there is exactly one pitchfork bifurcation from the zero state for each such pair included in the model.
These equilibria emerge at the Rayleigh number $\R_{mn}$, given by
\begin{equation} \label{eq:RLmn}
\R_{mn} := \left((mk)^2 + n^2\right)^3/(mk)^2.
\end{equation}
The HK4 model has one bifurcation from the zero state, at the Rayleigh number $\R_{11} = \left(k^2 + 1 \right)/k^2$, and no additional bifurcations occur. Each successive model in the hierarchy introduces one additional pitchfork bifurcation from the zero state. 
In some cases, especially in smaller models, these equilibria are direct analogues of the Lorenz equilibria for all values of $\R$. When this occurs we denote the corresponding equilibria by $L_{mn}$. For larger models, the equilibria emerging at $\R_{mn}$ can take a more complicated form, deviating form the Lorenz-like subspace due to nonlinear interactions between modes, as will be explored below. First, we examine the simpler case of the $L_{mn}$ equilibria.

The nonzero variables in the $L_{mn}$ equilibria are given by
\begin{equation} \label{eq:Lmn_simple}
\begin{gathered}
\psi_{mn} = \pm \tfrac{\sqrt{8}}{(mk)^2 + n^2} \sqrt{\R - \R_{mn}}, \qquad \theta_{mn} = \pm (-1)^{m+n} \tfrac{\sqrt{8}}{\R} \tfrac{(mk)^2 + n^2}{mk} \sqrt{\R - \R_{mn}}, \\ \theta_{0,2n} = \tfrac{1}{n \R} \left(\R - \R_{mn}\right).
\end{gathered}
\end{equation}
Each of the $L_{mn}$ equilibria can be related to the equilibria of the Lorenz equations by a linear change of variables after shifting the Rayleigh number by $\R_{mn} - \R_{11}$ to the corresponding bifurcation point. The value of $N$ at the $L_{mn}$ equilibria, computed from either \eqref{eq:N_gen} or \eqref{eq:N2_gen}, is 
\begin{equation} \label{eq:NLmn}
N_{L_{mn}} = 3 - 2 \R_{mn}/\R.
\end{equation}

The equilibria that bifurcate at $\R_{11}$ are of particular interest because they correspond to a pair of steady convection rolls that are globally attracting at onset. In the simplest case, these equilibria are called $L_{11}$ and are equivalent to the equilibria that emerge at the first instability of the Lorenz equations \cite{Lorenz63} and the HK8 model \cite{Olson2020}. The value of $\R_{11}$ is smallest when $k^2 = 1/2$, so we denote this value as the critical Rayleigh number, $\R_c$. 

In higher-order models of the HK hierarchy, the equilibria arising at $\R_{mn}$ often deviate from the subspace spanned by the variables of $L_{mn}$. This occurs due to pairing in the quadratic terms on the right-hand side of the ODEs. In the language of \S\ref{sec:constructModel}, the pairing occurs due to \emph{compatible triplets} that comprise terms in the sum $Q_{mn}^\psi$ and $Q_{mn}^\theta$.

Just as in the Lorenz-like case, the equilibria that emerge at $\R_{11}$ in larger models correspond to the onset of convection rolls when $k^2 = 1/2$, and arise due to an instability in the $\psi_{11}$--$\theta_{11}$ subspace. We generally call such equilibria the \emph{primary equilibria} of a given model in the HK hierarchy as a generalization of the $L_{11}$ branch. In the HK hierarchy, expressions for the primary equilibria first differ from the $L_{11}$ states in the HK14 system. The initial $L_{11}$ subspace ($\psi_{11}$, $\theta_{11}$, $\theta_{02}$) induces a perturbation in $\theta_{13}$ due to a term proportional to $\psi_{11} \theta_{02}$ in the $\theta_{13}$ ODE. As a result, $\theta_{13}$ must be nonzero along the primary branch. In turn, this activates the $\psi_{13}$ mode, due to the linear pairing that occurs between any $\psi$ and $\theta$ modes with equal wavenumber. Subsequent nonlinear pairing activates the $\theta_{04}$ and $\theta_{06}$ variables along the primary branch, resulting in equilibria that lie within a 7-dimensional subspace. The cascade of modes activating along the primary branch as described here is depicted in Figure \ref{fig:smallR}. This nonlinear pairing mechanism is observed for all models in the HK hierarchy with $M_i \geq 14$, since all such models include the mode $\theta_{13}$ that begins the cascade.

\begin{figure}[tp]
\centering
\includegraphics[scale=.5]{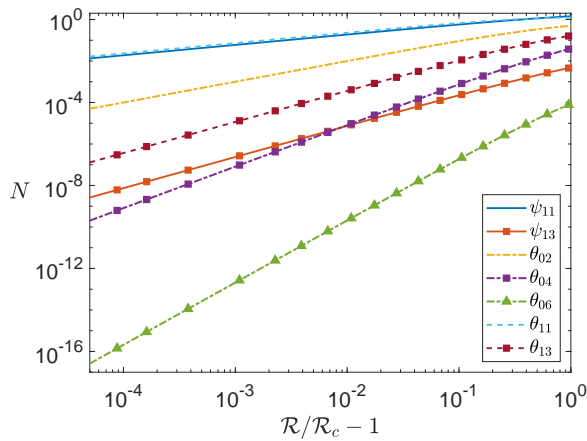}
\caption{Growth of the nonzero variables along the primary branch of the HK14 system near the pitchfork bifurcation from the zero state, with the model parameters fixed at $\sigma = 10$ and $k^2 = 1/2$. The difference in magnitude of the variables near the bifurcation point illustrates the cascade of modal activation from the initial subspace ($\psi_{11}$--$\theta_{11}$) to the $\theta_{06}$ mode that completes the fully developed subspace of the primary branch.}
\label{fig:smallR}
\end{figure}

A consequence of the above mechanism is that the primary equilibria only contain modes with even total wavenumber ($m + n$). This is because the interaction of two modes of even wavenumber in the nonlinear parts of \eqref{eq:psiTrunc}--\eqref{eq:thetaTrunc} can only excite a mode of even wavenumber (see \ref{sec:proofs} for details). Therefore, the primary branch cannot contain shear modes, since all such modes have odd total wavenumber under the phase convention described in \S\ref{sec:simplifyModel}. On the other hand, all modes of the form $\theta_{0,2n}$ are eventually activated on the primary branch for sufficiently large $M_i$. This augments the heat transport as the system size increases. It is possible that the heat transport in the PDE is also subject to a similar modal cascade as the Rayleigh number grows. The nonzero modes in the primary equilibria are summarized in Table \ref{tab:primaryEquil}.

The primary branches of equilibria for several reduced-order models are depicted in Figure \ref{fig:HKStream}. For each of these models, the temperature profiles display slightly unphysical behavior at $\R = 5\R_c$, including internal temperature maxima. These features are not present at the onset of convection and develop at some larger $\R$, indicating that the reduced model is not capturing the full physics of the Boussinesq equations.

In general, the expressions for the equilibria that emerge at $\R_{mn}$ deviate from \eqref{eq:Lmn_simple} for models large enough such that $\psi_{m,3n}$ and $\theta_{m,3n}$ are included. Similar to the cascade examined for the primary equilibria, nonlinear pairing between $\psi_{mn}$ and $\theta_{0,2n}$ occurs in the $\theta_{m,3n}$ equation, so that additional variables must be nonzero. We observe that the inclusion of additional modes generally enhances heat transport when compared to \eqref{eq:NLmn}. 

\begin{table}[!ht]
\centering
\ra{1.1}
\caption{Models that complete each shell in the HK hierarchy, along with the terms added to the horizontally-averaged version of the truncated Nusselt number \eqref{eq:N2_gen}, and the number of nonzero terms in the fully-developed primary branch of equilibria, $n_{primary}$.} \label{tab:primaryEquil}
\begin{tabular}{cccc}
\toprule
Shell & Largest model & New term in \eqref{eq:N2_gen} & $n_{primary}$ \\
\midrule
1 & HK4 & $\theta_{02}$ & 3 \\
2 & HK10 & $\theta_{04}$ & 3 \\
3 & HK18 & $\theta_{06}$ & 11 \\
4 & HK28 & $\theta_{08}$ & 11 \\
5 & HK40 & $\theta_{0,10}$ & 23 \\
6 & HK54 & $\theta_{0,12}$ & 23 \\
7 & HK70 & $\theta_{0,14}$ & 39 \\
8 & HK88 & $\theta_{0,16}$ & 39 \\
9 & HK108 & $\theta_{0,18}$ & 59 \\
10 & HK130 & $\theta_{0,20}$ & 59 \\
11 & HK154 & $\theta_{0,22}$ & 83 \\
\bottomrule
\end{tabular}
\end{table}

\begin{figure}[!htb]
\centering
\begin{tikzpicture}
\node at (0,0) {\includegraphics[scale=0.5,trim = 0 35 0 0, clip]{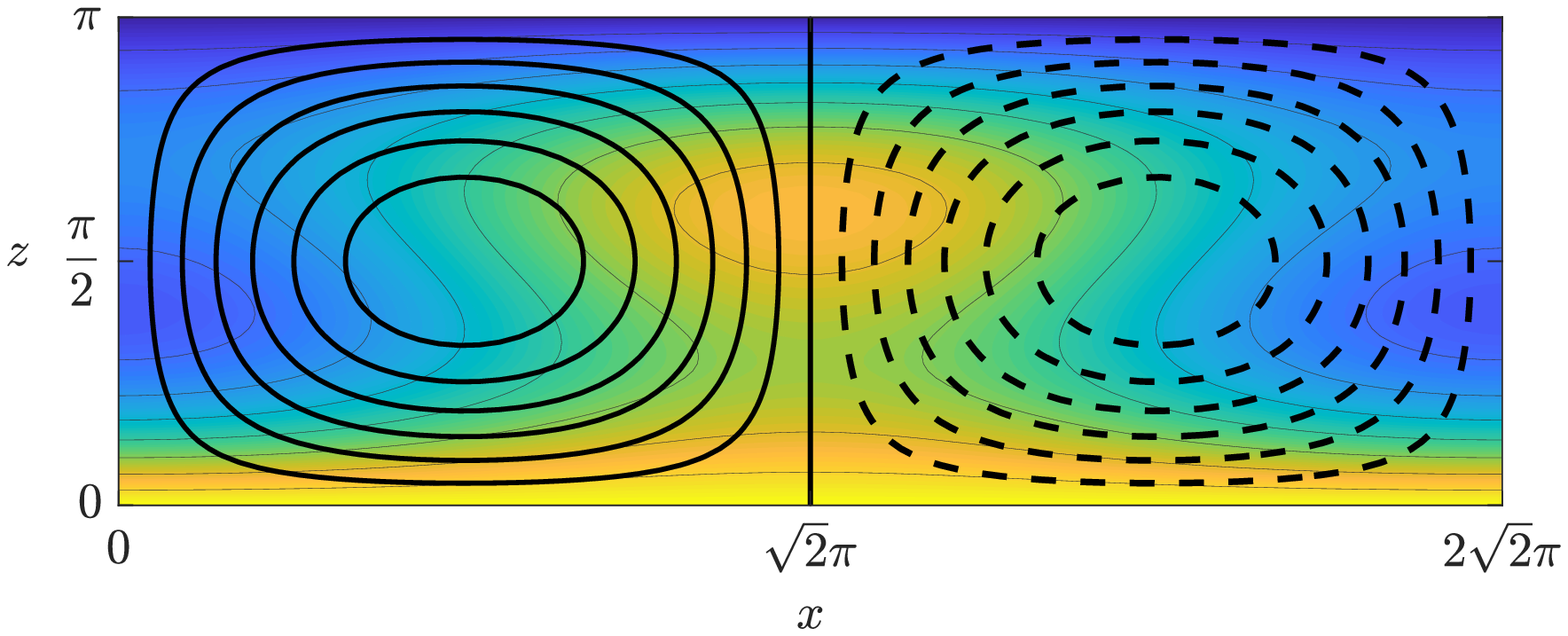}};
\node at (-5.1,1.3) {\small (a)};
\end{tikzpicture}
\begin{tikzpicture}
\node at (0,0) {\includegraphics[scale=0.5,trim = 0 35 0 10, clip]{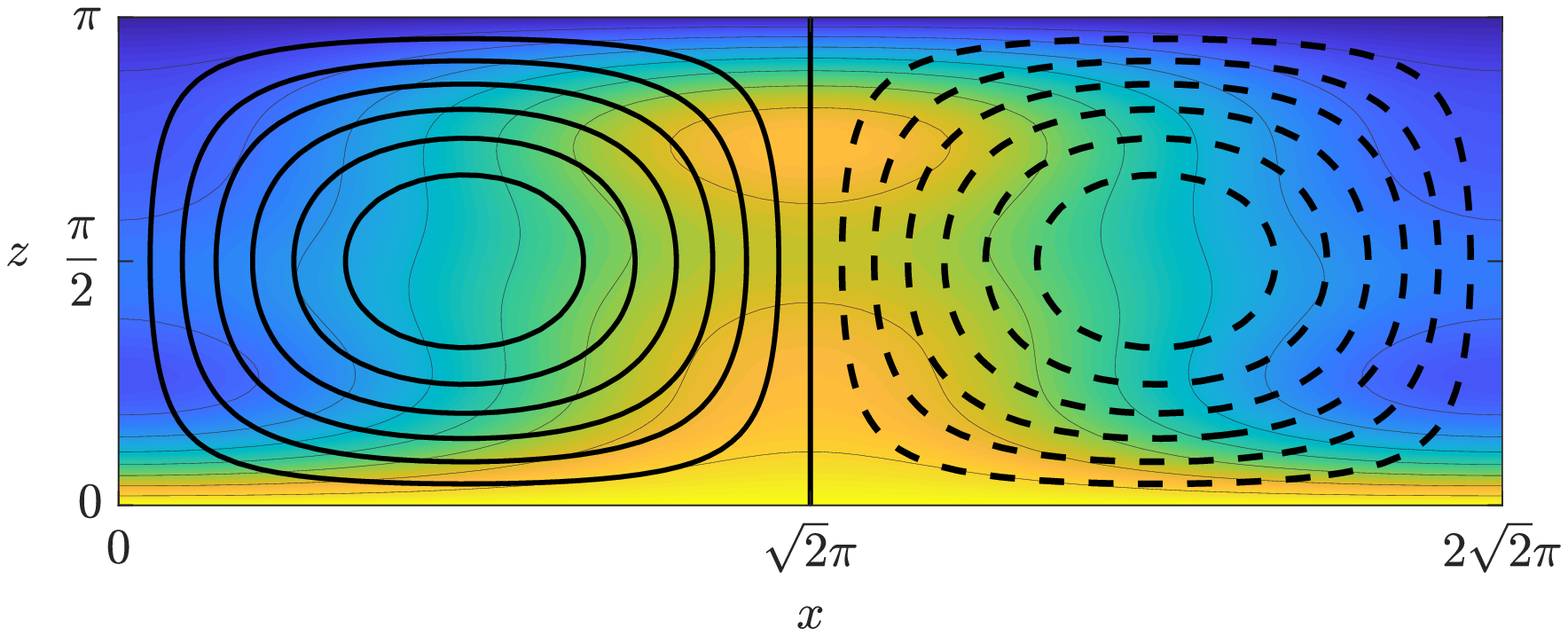}};
\node at (-5.1,1.2) {\small (b)};
\end{tikzpicture}
\begin{tikzpicture}
\node at (0,0) {\includegraphics[scale=0.5,trim = 0 35 0 10, clip]{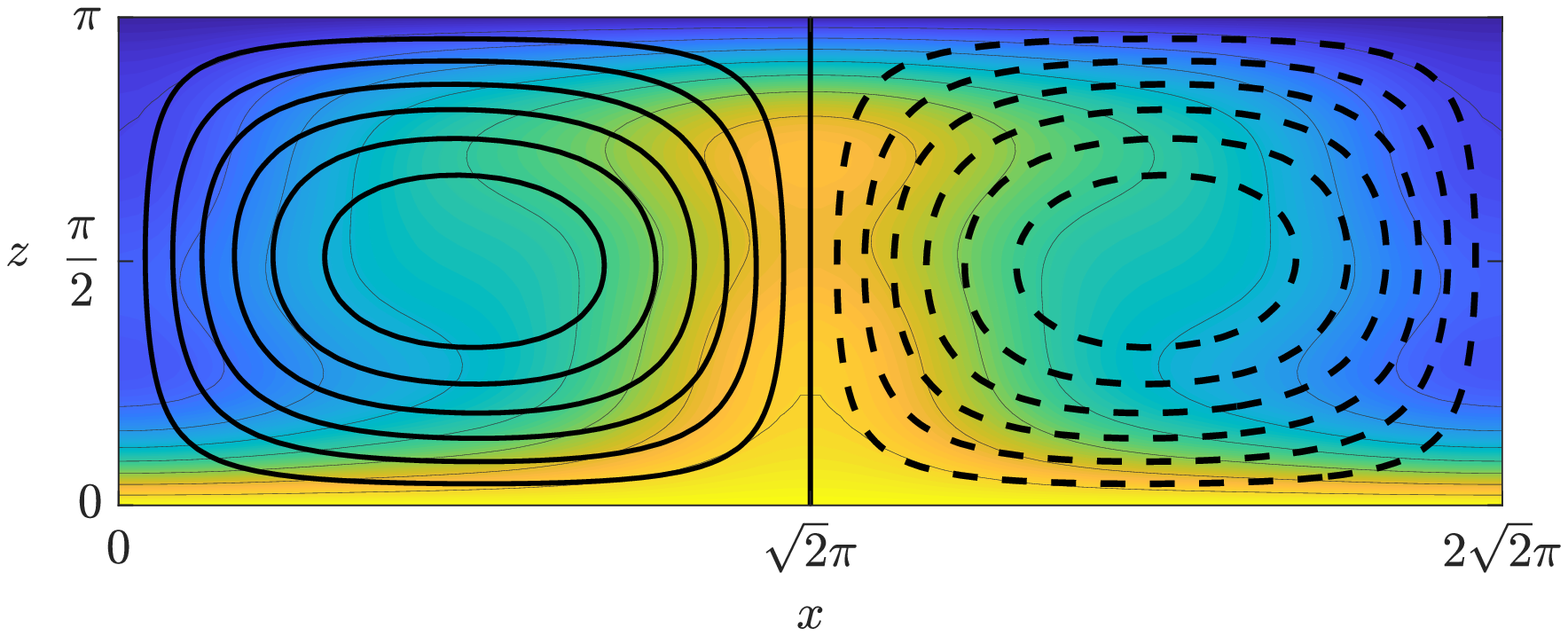}};
\node at (-5.1,1.2) {\small (c)};
\end{tikzpicture}
\begin{tikzpicture}
\node at (0,0) {\includegraphics[scale=0.5,trim = 0 0 0 10, clip]{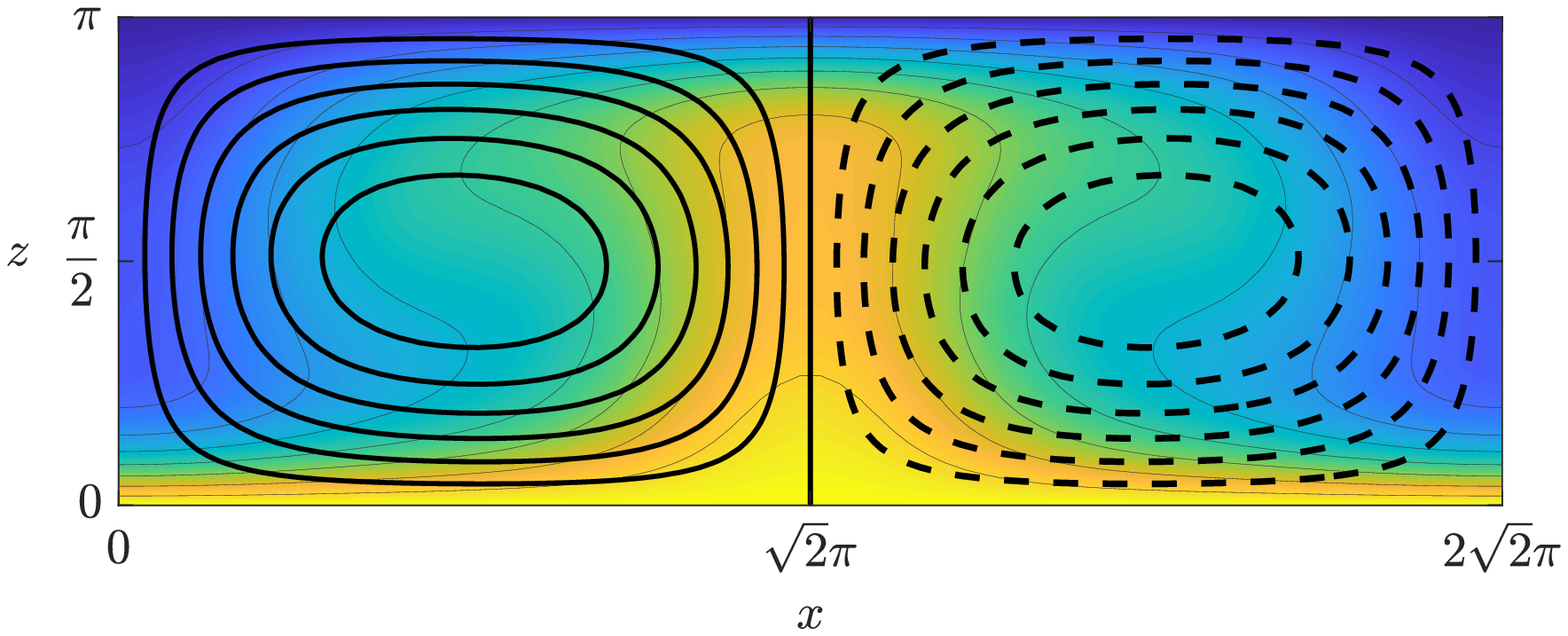}};
\node at (-5.1,1.4) {\small (d)};
\end{tikzpicture}
\caption{Temperature ($T$) contours for approximations of steady convection with modal amplitudes in the Galerkin expansion \eqref{eq:Galerkin} are given by the primary branches of equilibria. The plots above depict several reduced models at $k^2 = 1/2$, $\sigma = 10$ and $\R = 5 \R_c$: (a) the HK8 model (b) the HK14 model (c) the HK22 model and (d) the HK32 model. The temperature ranges from 1 (light) to zero (dark). Solid and dashed streamlines indicate lines of positive and negative vorticity, respectively. In each model, additional modes pair with the nonzero variables of the $L_{11}$ states due to nonlinear interactions between modes, enhancing heat transport across the domain.
} \label{fig:HKStream}
\end{figure}

We studied the bifurcation structure of the HK$M_i$ models in more detail with the numerical continuation package MATCONT \cite{Dhooge2003}. For these and subsequent computations within this section, the model parameters $k^2$ and $\sigma$ were fixed at $1/2$ and $10$ respectively. To improve numerical stability for all computations performed on the HK models, we scale the $\psi$ modes by $\R^{-1/2}$ and time by $\R^{1/2}$ (see \ref{sec:NumProc}). As a result, the Nusselt number for the scaled HK models is 
\begin{equation} \label{eq:N_HK}
N = 1 + \sum_{\mathclap{(0,2n) \in S_\theta}} (2n) \, \ov{\theta_{0,2n}}.
\end{equation}
With $M_i \leq 44$, we locate all branches of equilibria detectable for $\R/\R_c < 1000$, with $\R$ as the bifurcation parameter. We began by continuing all branches that bifurcate from the zero state at $\R_{mn}$. For these and each additional branch located, we continued the equilibria until $\R/\R_c$ was at least 1000, or until the curve terminated. This process was repeated for each equilibria stemming from any of the pitchfork bifurcations detected, until no additional branches of equilibria were found. The number of equilibria grows rapidly as $M_i$ increases; we detected nine nonzero equilibria for the HK14 model, while 127 equilibrium branches were detected for the HK40 system. Results for the HK10 and HK14 models are displayed in Figure \ref{fig:bfcs}. The primary equilibria of the HK14 model display appreciably greater heat transport than those of the HK10 model at all values of $\R$ past onset. This is due to the mechanism described above where the $L_{11}$ equilibria are augmented with additional modes.

\begin{figure}[tp]
\centering
\begin{tikzpicture}
\node at (0,0) {\includegraphics[scale=.5]{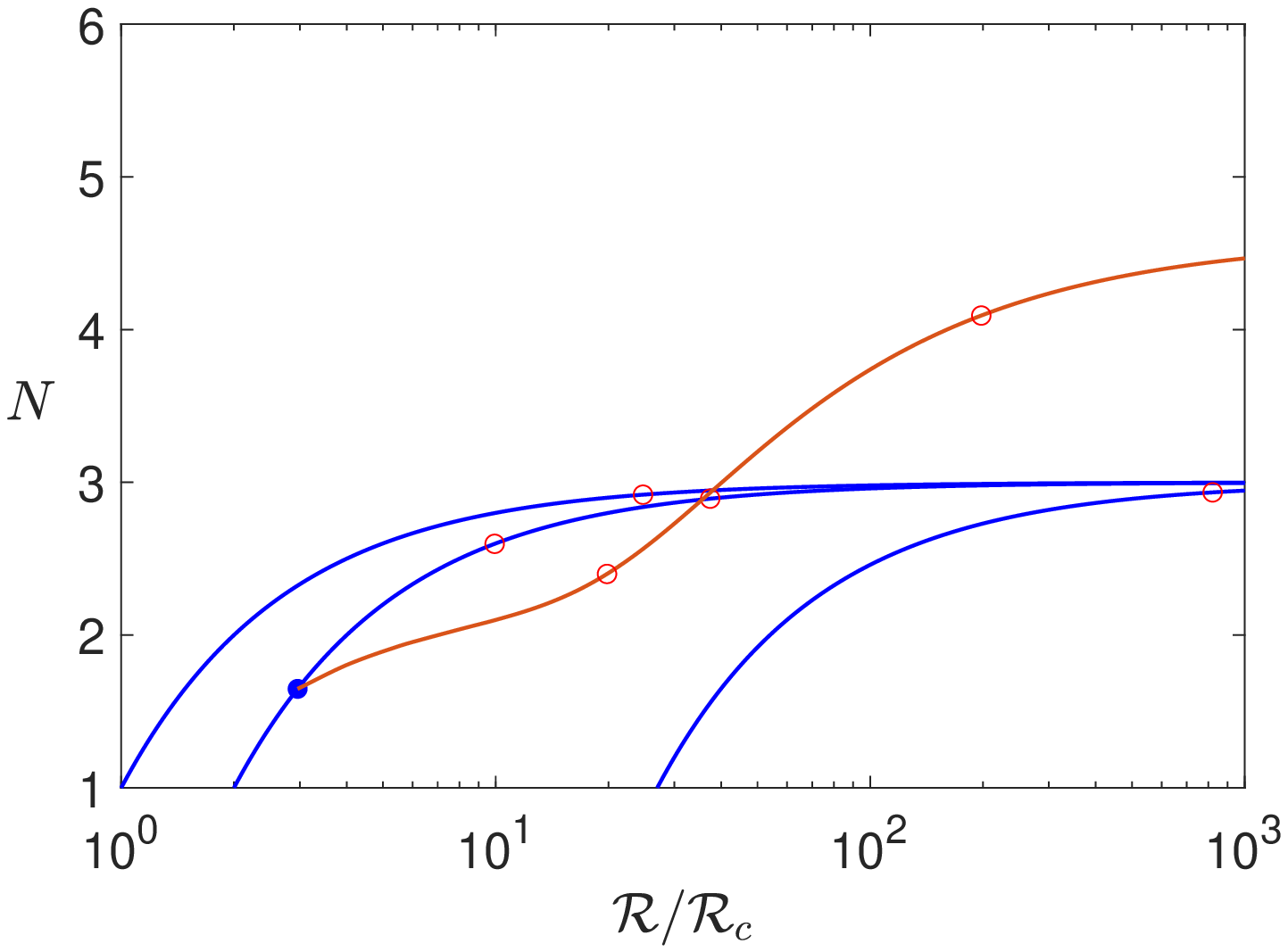}};
\node at (-2.3,2.2) {HK10};
\end{tikzpicture}
\begin{tikzpicture}
\node at (0,0) {\includegraphics[scale=.5, trim=19 0 0 0, clip]{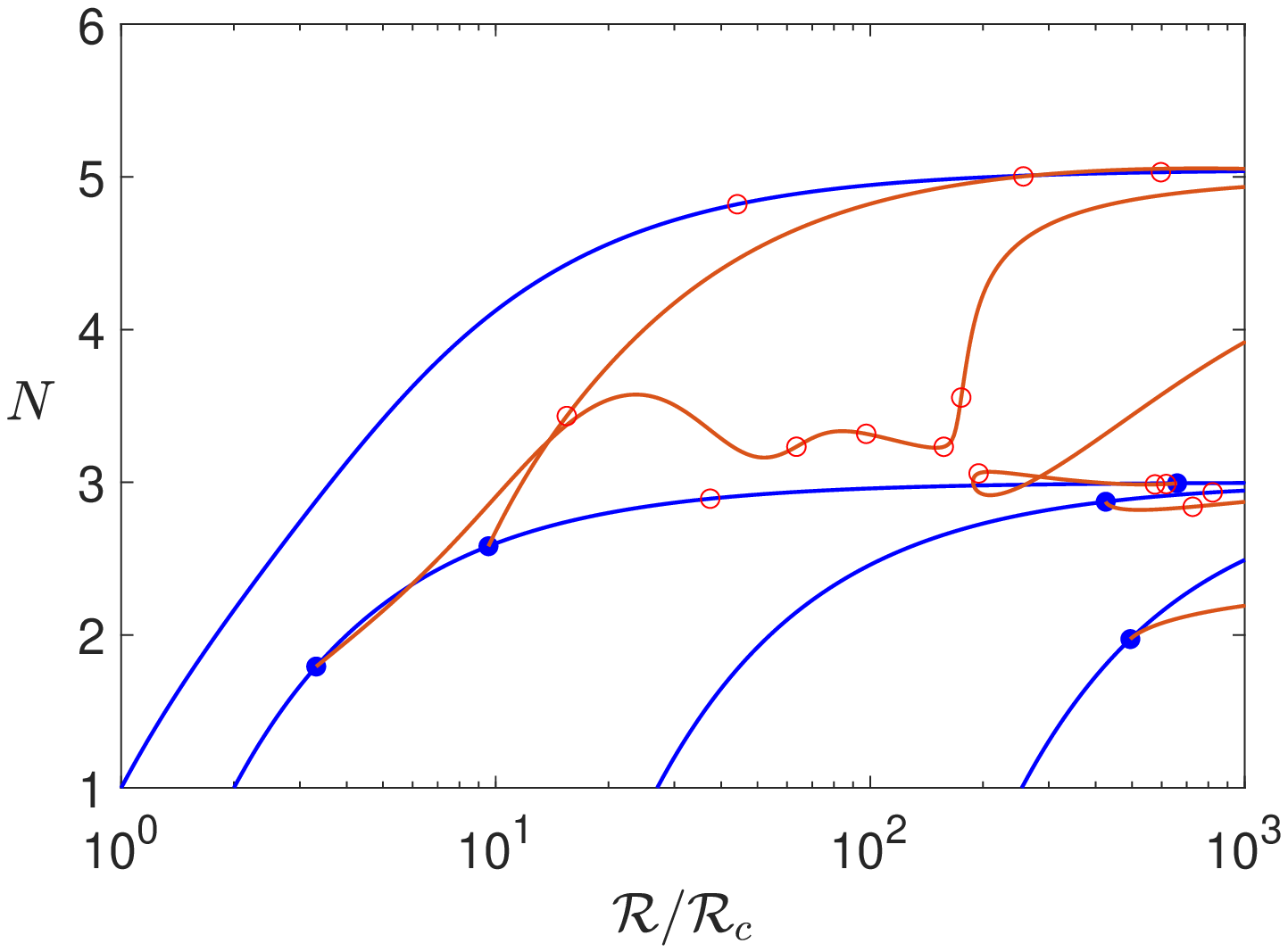}};
\node at (-2.5,2.2){HK14};
\end{tikzpicture}
\caption{Bifurcation diagrams for the HK10 and HK14 models, at $k^2 = 1/2$ and $\sigma = 10$. Filled circles indicate pitchfork bifurcations and open circles denote Hopf bifurcations. All curves were computed by numerical continuation with a resolution of approximately 0.1 in units of $\R/\R_c$.}  \label{fig:bfcs}
\end{figure}

To obtain candidates for the maximum heat transport for a given model, we maximized $N$ among all computed equilibria for each model with $M_i \leq 40$. In all cases, the primary equilibria are maximal from onset until some larger value of $\R$. When $M_i > 4$, the heat transport of the primary branch is eventually surpassed by that of some other equilibrium branch, but this only happens when $\R$ is well beyond the point where unphysical behavior is first observed, and may simply be an artifact of the truncation. The primary equilibria were compared to the analogous states of the Boussinesq equations---those arising from the first instability from the static state. We observe that the HK14 model predicts slightly larger heat transport than the PDE for some $\R$. This may occur as a result of the partially filled shell in the hierarchy, where only one Lorenz triple with $m+n=4$ is included in the truncation. For models that complete a shell (HK10, HK18, and so on), the value of $N$ along primary equilibria closely approximates Nu at corresponding PDE steady state for small $\R$. The interval of $\R$ where $N \approx \mathrm{Nu}$ increases as $M_i$ is raised.

The bifurcation structure of the HK models becomes more complex with larger $M_i$, and the number of equilibria and Hopf bifurcations rapidly increases with the dimension of the ODE. Therefore it is not practical to attempt to locate every equilibrium branch when the dimension becomes sufficiently large. For larger models, we consider only the primary branch of equilibria, and conjecture that these equilibria transport heat optimally at all physically relevant values of $\R$. The primary equilibria for several models are shown in Figure \ref{fig:bfc_comp_pde} alongside the Nusselt number of the analogous steady state of the PDE. 


\begin{figure}[ht]
\centering
\hspace{0.5cm} \begin{tikzpicture}
\node at (0,0) {\includegraphics{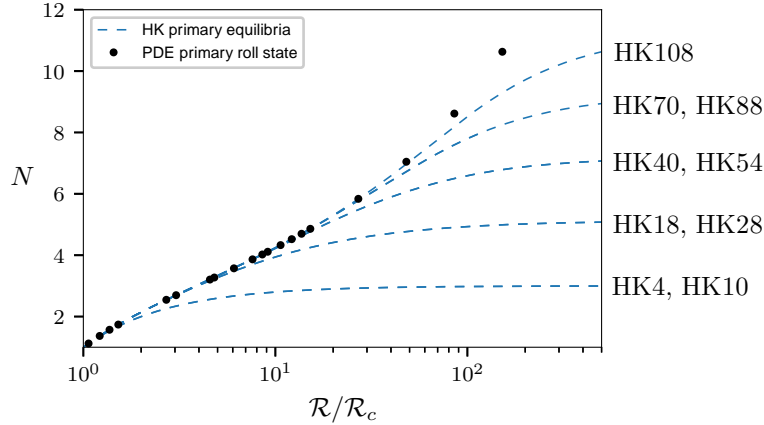} };
\node at (4.6, -1.25) {HK4, HK10};
\node at (4.7, -.42) {HK18, HK28};
\node at (4.7, .4) {HK40, HK54};
\node at (4.7, 1.16) {HK70, HK88};
\node at (4.25, 1.9) {HK108};
\end{tikzpicture}
\caption{Heat transport along the primary branch of equilibria found via numerical continuation for several selected truncated models that complete a shell in the hierarchy. Values of Nu at the equilibria arising from the first instability of the Boussinesq equations---analogues to the primary equilibria---are plotted for comparison. For each model with $M_i \leq 40$, the primary branches were determined to provide maximal heat transport among all equilibria until $\R/\R_c$ was greater than 20. Data for the PDE was computed by Baole Wen \cite{Wen2020}.} \label{fig:bfc_comp_pde}
\end{figure}

\subsection{Time integration of the HK ODEs} \label{sec:TimeDep}
For models in the HK hierarchy with $M_i \leq 44$, additional candidates for the maximal $N$ were obtained by directly computing the time average integral \eqref{eq:timeAvg} with $\Phi = N$. Numerical solutions were obtained by directly integrating the ODEs for $10^4$ to $10^5$ time units with the MATLAB solver \texttt{ode45} starting from randomly generated initial conditions within $[-1,1]^{M_i}$. The absolute and relative tolerances of the solver were set to $10^{-12}$ and $10^{-9}$, respectively. In cases where more than one stable solution was detected at a particular Rayleigh number, the maximum $N$ is computed among such solutions. The results are compared with the value of $N$ along equilibrium branches and sum-of-squares upper bounds in Figure \ref{fig:boundsHK}. 

When the Rayleigh number is sufficiently small, integrating the ODEs yields stable equilibria, while at larger $\R$ we identify attracting sets that are time-periodic or even chaotic. The heat transport along time-dependent trajectories is found to be smaller than that of the primary equilibrium branch whenever $\R$ is small enough to produce a meaningful comparison to the physics of the PDE. A more detailed study of the time integration of reduced-order models within a similar framework can be found in \cite{Park2021}.

\section{Sum-of-squares optimization} \label{sec:SOS}
Time averages of functions of dynamical variables are often of interest for nonlinear dynamical systems, more so than the value of the function at any particular instant in time. In recent years, a general technique has been developed to determine upper or lower bounds on time-averaged quantities for ordinary differential equations~\cite{Cherny2014}. Such results are global in the sense that they provide a bound on a given quantity over all solution trajectories with initial conditions in a given domain. These methods involve choosing an \emph{auxiliary function} defined on the state space of the ODE that facilitate proving the desired bound. Using auxiliary functions to prove bounds on time averages is reminiscent of the common technique of using Lyapunov functions to determine global stability properties for ODEs. Like the functions used in Lyapunov's method, auxiliary functions need not have any particular relationship to the system dynamics. The auxiliary function method has been applied for a variety of systems~\cite{Goluskin2018,Goluskin2019,Fantuzzi2016},  including modal approximations of PDEs such as the models described in \S\ref{sec:Galerkin}.
We present an overview of the auxiliary function method in \S\ref{sec:TimeAvg}. The application to polynomial dynamical systems is discussed in \S\ref{sec:PolySOS}, allowing bounds to be computed numerically with techniques of convex optimization.

\subsection{Maximal time averages for ODEs} \label{sec:TimeAvg}
We first present the auxiliary function method applied to a generic well-posed autonomous ODE $\dot{\bx} = \mathbf{f}(\bx)$ governing the dynamics of $\bx \in \mathbb{R}^n$. Here $\dot{\bx}$ denotes the time derivative of $\bx$, and we assume $\mathbf{f}: \mathbb{R}^n \to \mathbb{R}^n$ is continuously differentiable and that all solutions to the ODE are bounded forward in time. Each solution of the ODE is a trajectory $\bx(t)$ corresponding to the initial condition $\bx(0) = \bx_0$. The long-time average of a continuous scalar quantity $\Phi(\bx)$ along the trajectory $\bx(t)$ is given by
\begin{equation} \label{eq:timeAvgPhi}
\ov{\Phi}(\bx_0) := \lim_{\tau \to \infty} \frac{1}{\tau} \int_0^\tau \Phi(\bx(t))\, \rm{d}t.
\end{equation}

Computing~\eqref{eq:timeAvgPhi} exactly is only possible when trajectories of the ODE are known exactly, such as at the equilibria. In practice, when formulae for the relevant trajectories are not explicitly known, the time average may be estimated by numerically integrating the ODE over a sufficiently large time and using the result to approximate the limit in~\eqref{eq:timeAvgPhi}. Estimates obtained in this way may converge slowly, and are restricted to only sufficiently stable trajectories stemming from a set of chosen initial conditions. It is generally impossible to perform this computation over all relevant initial conditions, especially in systems exhibiting chaotic behavior.

Auxiliary functions~\cite{Cherny2014} allow another way to obtain information about time-averaged quantities for dynamical systems, without selecting a particular trajectory. The auxiliary function method provides bounds that are global in the sense that the bounds on~\eqref{eq:timeAvgPhi} hold over all trajectories (or equivalently, over all initial conditions) of the ODE. Accordingly, let $\ov{\Phi}^*$ be defined by
\begin{equation} \label{eq:PhiOpt}
\ov{\Phi}^* := \sup_{\bx_0 \in \mathbb{R}^n} \ov{\Phi}(\bx_0),
\end{equation}
and seek a global upper bound $U \in \field{R}$ so that $U \geq \ov{\Phi}^*$. While $U$ must be independent of the choice of trajectory, it may depend on the values of various model parameters. For example, upper bounds on $N$ will generally depend on $\R$, $\sigma$, and $k$. 

A global upper bound on $\ov{\Phi}^*$ could be constructed by computing the pointwise supremum of $\Phi$ \cite{Tobasco2018}:
\begin{equation}
\ov{\Phi}^* \leq \sup_{\bx \in \mathbb{R}^n} \Phi(\bx).
\end{equation}
However, in practice this will almost always produce bounds that are much larger than $\ov{\Phi}^*$, and will not produce a meaningful result unless $\Phi$ is bounded on $\mathbb{R}^n$. More useful bounds can be determined by introducing a continuously differentiable auxiliary function ${V:\mathbb R^n\to\mathbb R}$. For any such function, the quantity ${\overline{\mathbf f\cdot\nabla V}}$ vanishes along bounded trajectories of the ODE, since \cite{Goluskin2018}:
\begin{align}
\label{eq:fGradV}
\overline{\mathbf f(\mathbf x(t))\cdot\nabla V(\mathbf x(t))} =
\ov{\tfrac{\rm d}{{\rm d}t} V(\mathbf{x}(t))}
= \lim_{\tau \to \infty} \frac{1}{\tau} \Big[V(\mathbf x(\tau)) - V(\mathbf x(0))\Big] = 0.
\end{align}
This identity implies that given any initial condition $\mathbf{x}_0 \in \mathbb{R}^n$ and $V \in C^1$,
\begin{equation}
\label{eq:bound1}
\overline\Phi = \ov{\Phi+\mathbf f\cdot \nabla V}\le
\sup_{\mathbf x\in \mathbb{R}^n}\left[\Phi(\bx)+\mathbf f(\bx)\cdot\nabla V(\bx)\right].
\end{equation}
The quantity on the right-hand side of \eqref{eq:bound1} can be computed or estimated without solving the ODE, and the resulting supremum is finite for properly chosen $V$. Since the choice of initial condition in \eqref{eq:bound1} was arbitrary, this also provides a bound on $\ov{\Phi}^*$. Making the optimal choice of auxiliary function results in the upper bound: 
\begin{equation}
\label{eq: bound2}
\ov\Phi^* \le \inf_{V\in C^1}
\sup_{\mathbf x\in \mathbb{R}^n}\left[\Phi(\bx)+\mathbf f(\bx)\cdot\nabla V(\bx)\right],
\end{equation}
where $C^1$ denotes the class of continuously differentiable functions.

It was proved in~\cite{Tobasco2018} that for all bounded well-posed ODEs and continuous $\Phi(\bx)$, equality holds in \eqref{eq: bound2} with the optimization performed over a compact domain containing the attracting region of the ODE. Their result guarantees the existence of an auxiliary function (or sequence thereof) that yield arbitrarily sharp bounds on $\ov{\Phi}^*$. In practice, the infimum in~\eqref{eq: bound2} can often be attained \cite{Olson2020}. A convenient way to express~\eqref{eq: bound2} is to define a function $S(\bx)$ as
\begin{equation}
\label{eq:S}
S(\bx) := U -  \Phi(\bx) - \mathbf{f}(\bx) \cdot \nabla V(\bx).
\end{equation}
Then, an upper bound is implied by the nonnegativity of $S$, and the optimization problem~\eqref{eq: bound2} can be expressed as \cite{Tobasco2018}
\begin{equation}
\label{eq: bound3}
\ov\Phi^* = \inf_{\substack{V\in C^1\\[2pt]S\ge0}} U,
\end{equation}
where $S\ge0$ must hold for all $\bx \in \mathbb{R}^n$. The equality \eqref{eq: bound3} means that for every $U$ that is a valid upper bound on $\ov\Phi^*$, there exists a sequence of auxiliary functions certifying this bound. The challenge is to construct such a $V$ and verify that indeed $S\ge0$.

\subsection{Polynomial dynamical systems and sum-of-squares optimization} \label{sec:PolySOS}

While \eqref{eq: bound3} provides a sharp bound on $\ov{\Phi}^*$, determining the optimal auxiliary function is intractable in general because verifying that a polynomial is non-negative over a subset of $\mathbb{R}^n$ is generally NP-hard \cite{Murty1987} in both the degree and dimension of the polynomial. When the right-hand side of the ODE, $\mathbf{f}$, as well as the quantity $\Phi$ are each polynomial in $\bx$, the problem can be made tractable by restricting the class of auxiliary functions \cite{Goluskin2018}. The first step is to let $V$ be a polynomial of degree no larger than $d$, giving an optimization problem over the finite-dimensional set $\mathbb P_{n,d}$ of degree-$d$ polynomials in $n$ variables. Optimization over this smaller set sometimes gives a bound strictly larger than $\ov\Phi^*$, but in practice the bound converges quickly with increasing $d$. Next, the non-negativity constraint can be relaxed to the requirement that the polynomial $S$ be equal to a sum-of-squares (SOS) of other polynomial terms~\cite{Lasserre2001, Nesterov2000, Parrilo2000}. This stronger constraint ensures $S$ is non-negative over $\mathbb{R}^n$. 
Techniques of polynomial optimization have been applied to prove global stability results by constructing Lyapunov functions~\cite{Parrilo2000}, to identify the region of attraction for ODEs~\cite{Parrilo2003}, and to determine global bounds on time-averaged quantities~\cite{Cherny2014}.
One benefit of using an SOS constraint is that deciding whether a polynomial is in $\Sigma_{n,d}$ can be performed in polynomial time in both $n$ and $d$. An efficient algorithm for this purpose was developed in~\cite{Powers1998}, based on theoretical work on SOS polynomials by Shor~\cite{Shor1988,Shor1998}. 

If $S$ is assumed to be SOS, then for each fixed degree $d$, the upper bound from the resulting polynomial optimization problem is~\cite{Cherny2014, Fantuzzi2016, Goluskin2018}:
\begin{equation}
\label{eq:Opt}
\ov\Phi^* \leq U^*_d := \min_{V\in\mathbb{P}_{n,d}}~U \quad
\rm{s.t.} \quad S\in\Sigma_{n,d}.
\end{equation}
In this form, the polynomial optimization is computationally tractable for sufficiently small $n$ and $d$. Computations for SOS optimization were performed in this work for ODE models up to $n = 208$ when $d = 2$, and up to $n = 40$ when $d = 4$. If SOS optimization with a chosen degree $d$ does not yield a sharp upper bound, the bound generally improves at the expense of computational cost. 
In practice, bounds produced by SOS optimization often converge rapidly as the degree of $V$ is raised~\cite{Fantuzzi2016, Goluskin2018, Olson2020}. Optimization problems with SOS constraints are convex, and can be recast as a type of conic optimization problem called a semidefinite program (SDP). Most modern SOS algorithms utilize the Gram matrix form~\cite{Choi1994}, wherein the polynomial $S$ is represented as:
\begin{equation} \label{eq:Gram0}
S = \mathbf{b}^T \Q \mathbf{b},
\end{equation} 
for some vector of polynomial basis functions $\bf{b}$. 
It can be shown that $S \in \Sigma_{n,d}$ if and only if there exists a basis vector $\mathbf{b}$ such that the corresponding Gram matrix is symmetric positive semidefinite, written $\Q \succeq 0$~\cite{Powers1998}. 

Reformulating \eqref{eq:Opt} using the constraint \eqref{eq:Gram0} with polynomial basis $\mathbf{b}$ results in the convex optimization problem
\begin{equation}
\label{eq:SDP}
\ov\Phi^* \leq U^*_d := \min_{V\in\mathbb{P}_{n,d}}~U \quad \rm{s.t.} \quad
\begin{array}[t]{l}
S  = \mathbf{b}^\mathsf{T} \Q \mathbf{b} \\
\Q\succeq 0.
\end{array}
\end{equation}
In \eqref{eq:SDP}, affine constraints on the entries of $\Q$ are provided by the equality $S = \mathbf{b}^T \Q \mathbf{b}$, and $\Q \succeq 0$ defines a semidefinite constraint. Together, these constraints characterize \eqref{eq:SDP} as an SDP \cite{Boyd2004}. The polynomial $S$ \eqref{eq:S} includes all of the decision variables---the constant $U$ and the coefficients of the polynomial ansatz $V$---that are determined in solving the SDP. Efficient algorithms exist for solving SDPs, and have been employed in various studies on SOS optimization~\cite{Fantuzzi2016, Goluskin2018, Olson2020}. The result of applying these algorithms to models in the HK hierarchy are presented in \S\ref{sec:UB}. 

\section{Upper bounds on N} \label{sec:UB}

We apply the auxiliary function method with the sum-of-squares method introduced in \S\ref{sec:SOS} to determine upper bounds on $N^*$ for models in the HK hierarchy. Suppose we wish to produce bounds on the model with $M_i$ modes, whose $\psi$ and $\theta$ modes correspond to the ordered pairs in the sets $S_\psi$ and $S_\theta$, respectively. In the nomenclature introduced in \S\ref{sec:SOS}, the vector $\bf{x}$ consists of all modes in the HK$M_i$ model, and $\bf{f}$ is the right-hand side of the corresponding system of ODEs, with formulae constructed from \eqref{eq:psiTrunc}--\eqref{eq:thetaTrunc}. Next, let the function $\Phi$ be given by
\begin{equation} \label{eq:phiDef}
\Phi := 1 + \sum_{\mathclap{(0,2n) \in S_\theta}} (2n)\, \theta_{0, 2n},
\end{equation} 
whose time average is the truncated version of the Nusselt number for this model \eqref{eq:N_HK}. We seek $U^*_d$: the minimum upper bound that can be proved with degree $d$ auxiliary functions. In general $U^*_d$ depends on both the chosen model and the parameters $\R$, $\sigma$, and $k$. For bounds computed in this work, the number of modes is fixed in each individual SDP computation. Upper bounds are constructed using the optimization solver MOSEK \cite{mosek} (version 9.0.98). The toolbox YALMIP \cite{Lofberg2004, Lofberg2009} (version 20190425) is used to formulate the optimization problem in the form of \eqref{eq: bound3} and pass the problem to the solver. The upper bounds presented in this work were computed using a 3.0 GHz Xeon processor. Further computational details are presented in \ref{sec:NumProc}.

\subsection{Upper bounds with $k^2 = 1/2$, $\sigma = 10$} \label{sec:UBstd}

Bounds are constructed for HK models with $\sigma = 10$ and $k^2 = 1/2$, using degree four $V$ when $M_i \leq 40$, and using degree two $V$ when $M_i \leq 208$. Sum-of-squares upper bounds are plotted for several models in Figure \ref{fig:boundsHK}. In each case, $U^*_d$ increases with $\R$, growing rapidly at first and eventually leveling off as $\R$ is raised. The primary branch of equilibria saturates the upper bounds when the Rayleigh number is slightly larger than $\R_{11}$. Secondary equilibria emerging from the primary branch can saturate the bounds at larger $\R$ in some cases, such as the HK8 system analyzed in detail in \cite{Olson2020}. Whenever $\R$ is sufficiently small to allow comparison with the PDE, primary equilibria appear to saturate the bound.

The auxiliary function method provides sharp or nearly sharp bounds in the largest range of $\R$ when $M_i = 14$. This model is the first one whose primary equilibria deviate from the form $\eqref{eq:Lmn_simple}$, with seven nonzero modes when $\R > \R_{11}$, resulting in larger heat transport than the Lorenz equilibria. As such, the primary equilibria for HK14 are maximal for a larger range of Rayleigh number than smaller HK models, but are still simple enough to admit sharp bounds with $V$ of low degree.

When $M_i$ is increased for fixed $\R$, the bound $U^*_4$ increases noticeably upon progression to the HK14 and HK32 models. 
This appears to be caused primarily by the enhancement of heat transport that occurs when progression to the next model results in the pairing of modes with larger wavenumber with the nonzero modes of the primary branch of equilibria. Additional nonlinear modal pairing in the primary branch occurs only when the total wavenumber $(m+n)$ of the shell is even, since modes of odd total wavenumber do not pair with the primary equilibria. A consequence of this pairing mechanism is that the shear modes are identically zero at the primary equilibria. Similarly, in numerical simulations of the full PDE, zonal flow has been observed to decrease the time-averaged heat transport \cite{Goluskin2014}, and hence the corresponding shear modes must be zero in order to produce optimal heat transport.
\begin{figure}[!tbph]
\centering
\begin{tikzpicture}
\node at (0,0) {\includegraphics[scale=.5, trim=20 0 30 0, clip] {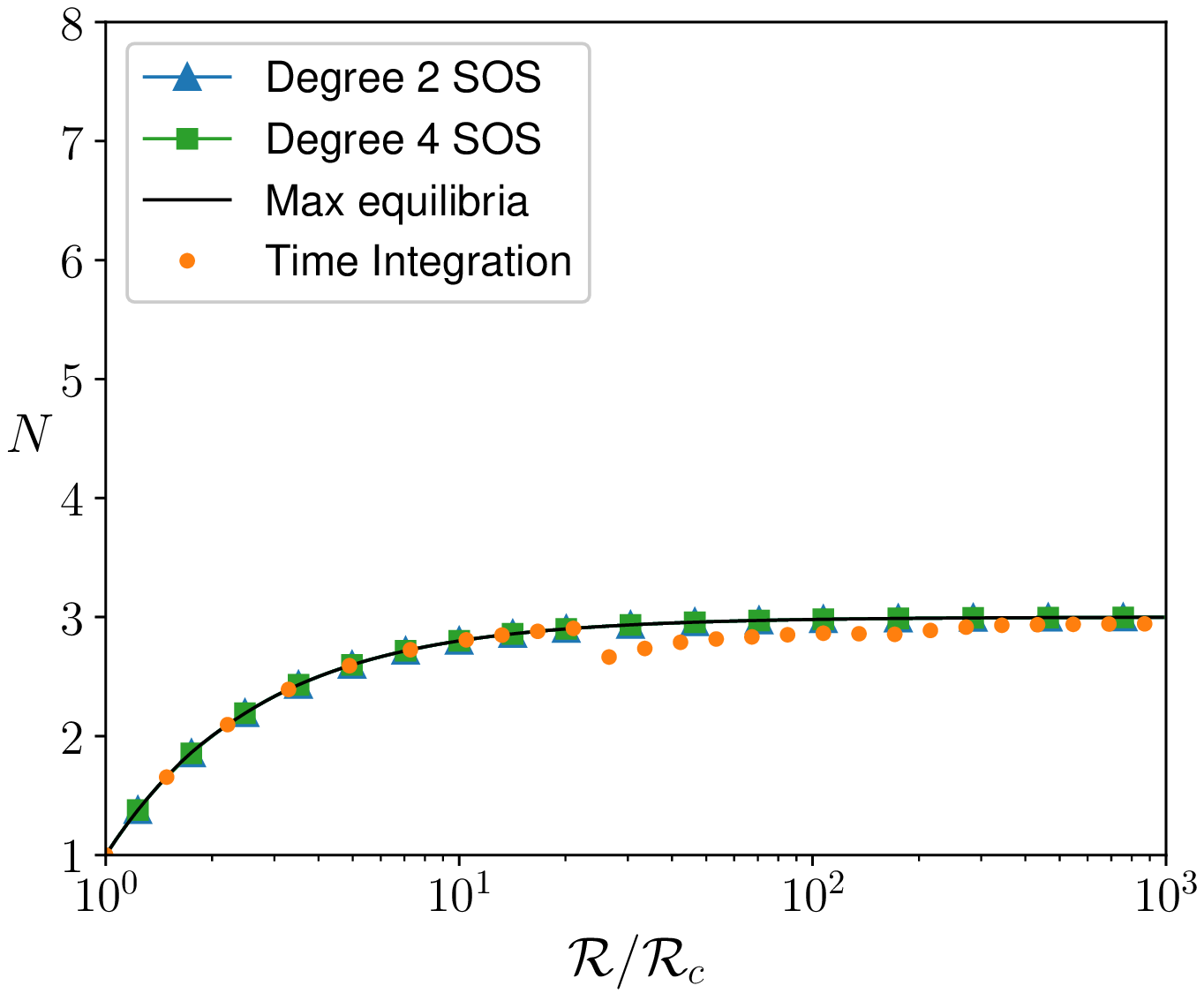}};
\node at (2.6,-2) {\small HK4};
\end{tikzpicture}
\begin{tikzpicture}
\node at (0,0) {\includegraphics[scale=.5, trim=20 0 30 0, clip] {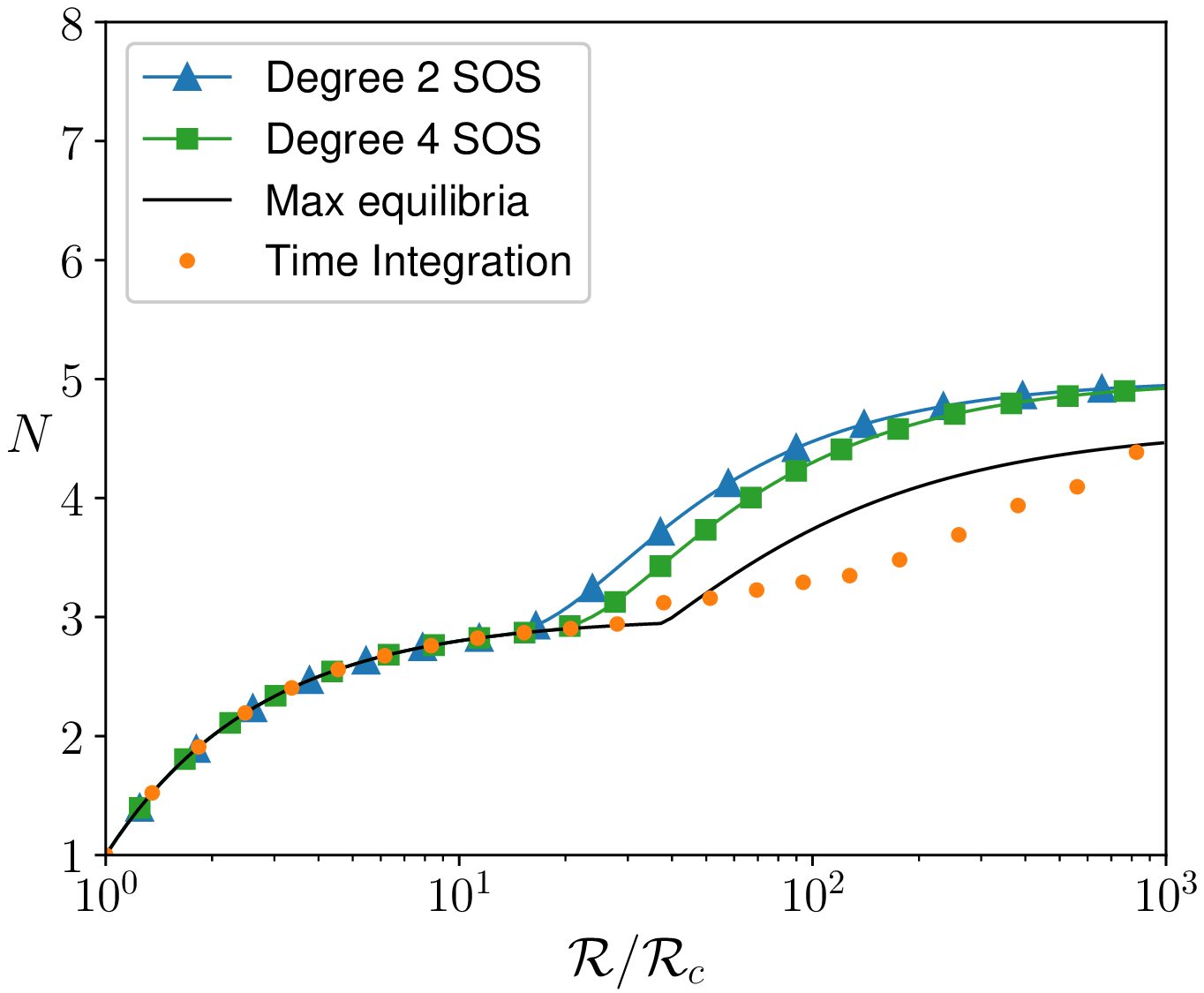}};
\node at (2.6,-2) {\small HK10};
\end{tikzpicture}\\
\begin{tikzpicture} 
\node at (0,0) {\includegraphics[scale=.5, trim=20 0 30 0, clip] {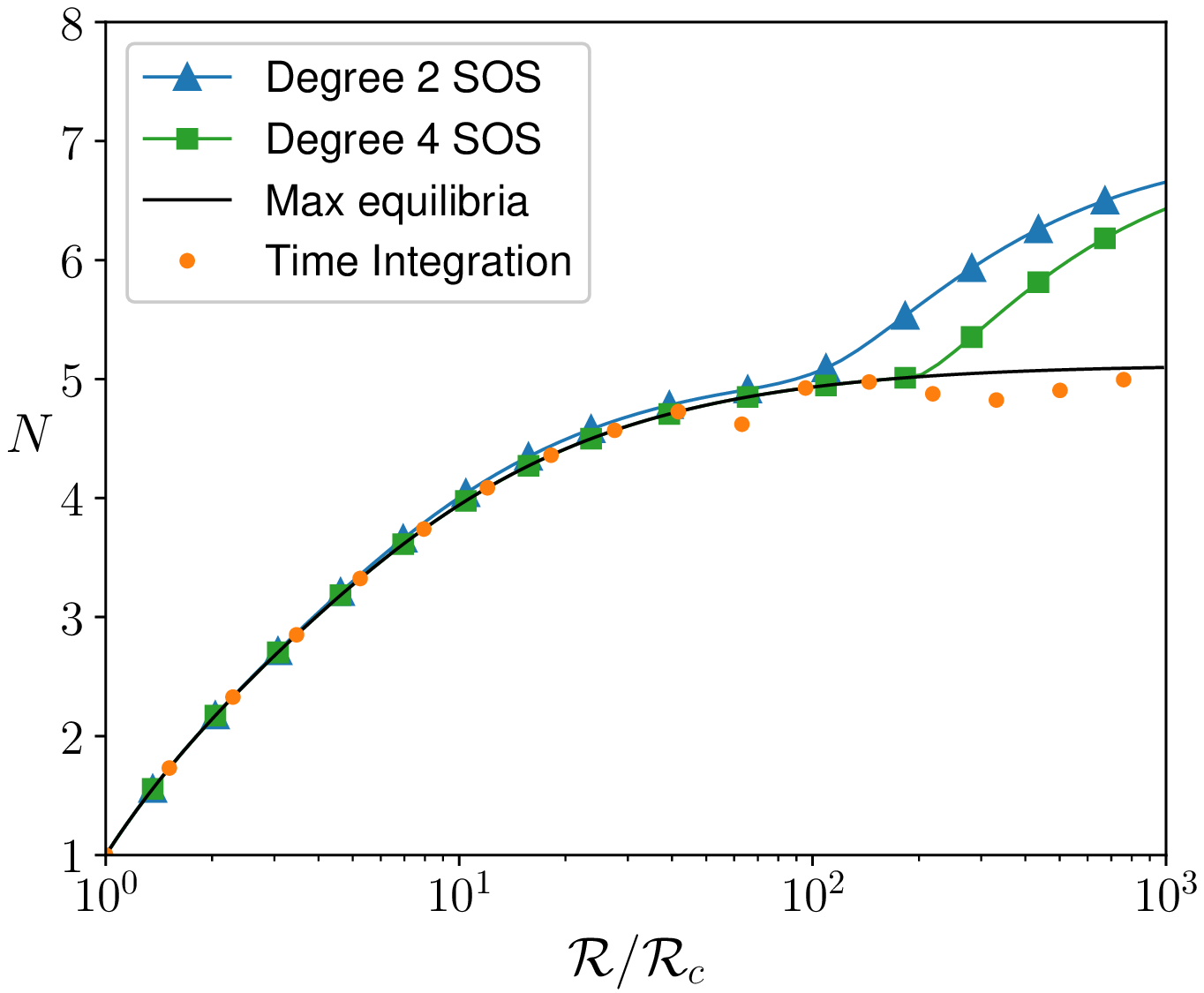}};
\node at (2.6,-2) {\small HK14};
\end{tikzpicture}
\begin{tikzpicture}
\node at (0,0) {\includegraphics[scale=.5, trim=20 0 30 0, clip] {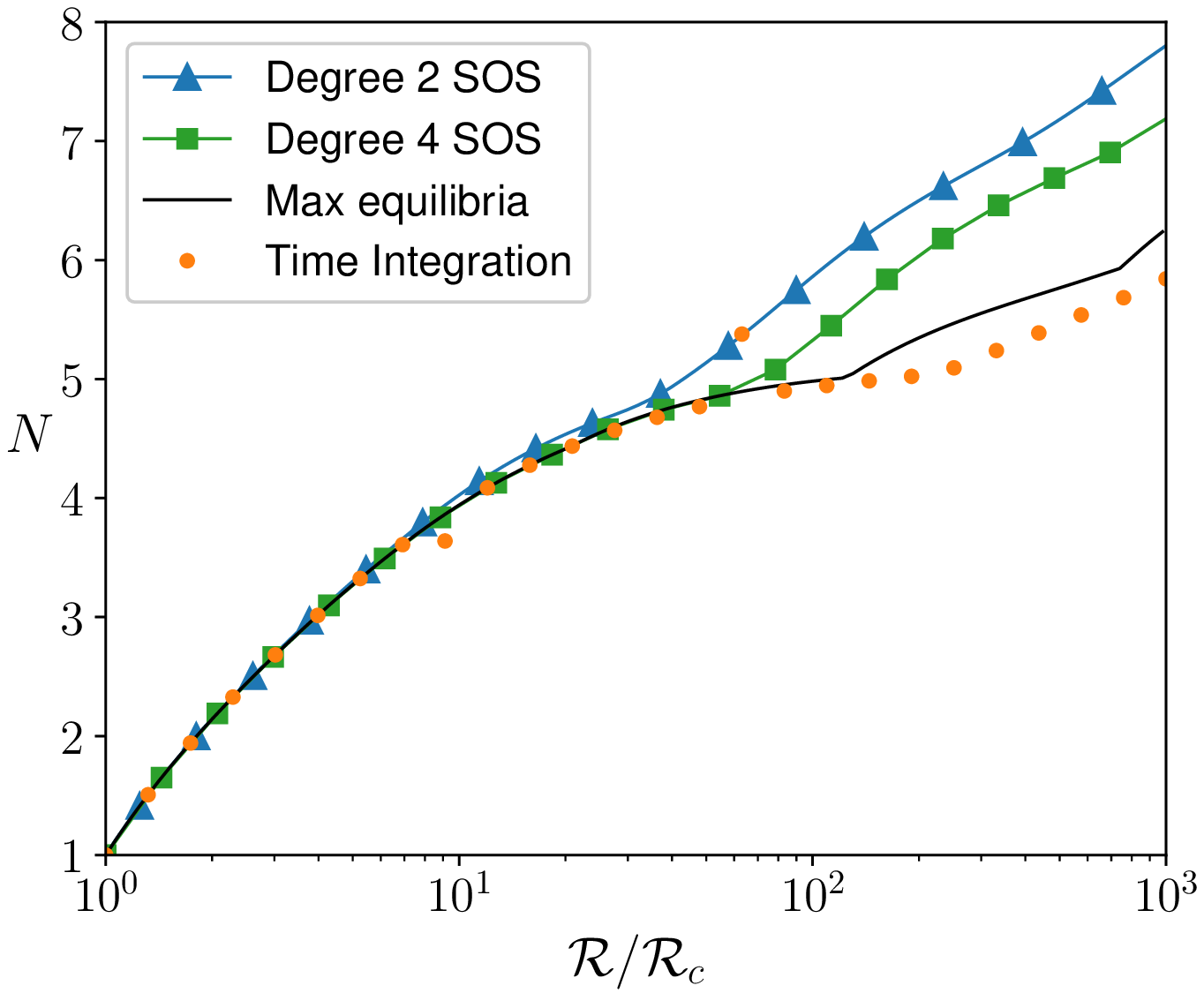}};
\node at (2.6,-2) {\small HK28};
\end{tikzpicture}
\caption{Upper bounds on $N$ obtained with the auxiliary function method for a few selected models in the HK hierarchy with degree two and degree four auxiliary functions. The maximum value of the Nusselt number obtained from numerical continuation and direct integration of the ODEs are shown for comparison.} \label{fig:boundsHK}
\end{figure}

SDPs with degree two $V$ provide more conservative upper bounds, but are less computationally taxing to compute, allowing bounds to be constructed for much larger systems. Such bounds are displayed in Figure \ref{fig:Deg2_large} for several models in the hierarchy up to HK208. The selected models are those that complete their respective shell in the HK hierarchy, such that progressing to the next model would require an additional temperature mode with a horizontal wavenumber of zero. Each new shell adds a term to \eqref{eq:phiDef}, resulting in a corresponding jump in the upper bound on $N^*$ at a given parameter combination. 

Suppose $\R$ is held fixed and the number of modes $M_i$ is increased. If the values of $N$ converge with increasing $M_i$, then the corresponding limit would be an upper bound for all models in the HK hierarchy. We call this upper bound Nu$^*$, since we expect that this would also provide an upper bound on Nu for the full PDE. With $M_i \leq 154$, the upper bound appears to be fully converged to Nu$^*$ up to about $\R = 30 \R_c$. For all $\R$ beyond this point, there exists a gap between the upper bounds for all models from different shells. As the Rayleigh number approaches infinity, the difference in $N^*$ when progressing to the next shell approaches two. For example, in the HK10 model, the large-$\R$ limit of the degree two upper bound is five, while the corresponding limit applied to the HK18 model is seven. This pattern continues, and in the HK154 model, this limit is 27.



\begin{figure}[!htb]
\centering
\includegraphics[scale=.5]{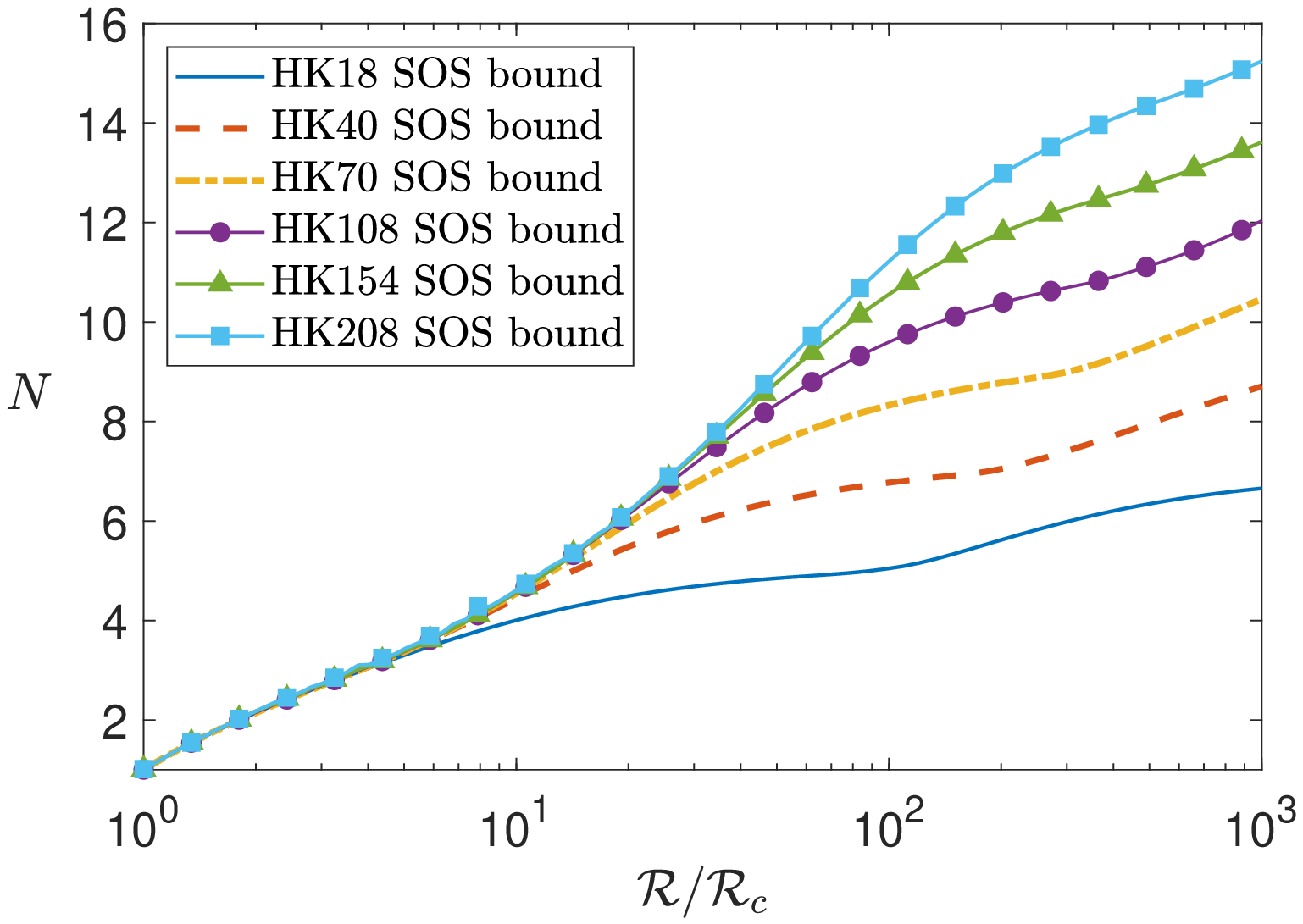}
\caption{Upper bounds on $N$ computed with degree two auxiliary functions for several models with $M_i \leq 208$ with parameter values $k^2 = 1/2$, $\sigma = 10$. The value of $N^*$ appears converged in the limit of increasing $M_i$ when $\R \leq 30 \R_c$.}
\label{fig:Deg2_large}
\end{figure}

\subsection{Upper bounds at optimal $k$}

In the preceding discussion, we analyzed the upper bounds at fixed $k$, namely the standard case when $k^2 = 1/2$. Here we instead maximize the value of $N$ over all possible $k$, resulting in an upper bound that holds for all domain aspect ratios. This is of particular interest because the full PDE for Rayleigh--B\'enard convection admits steady solutions of all horizontal periods when the Rayleigh number is sufficiently large. Specifically, we fix $\sigma$ and seek the solution $N^*_{k^*}$ to the optimization problem
\begin{equation}
N^*_{k^*}(\R,\sigma) := \sup_{k > 0} N^*(k,\R,\sigma),
\end{equation}
where $k^*$ is the value of $k$ that maximizes the upper bound. To estimate $N^*_{k^*}$, we use the MATLAB bounded optimization tool \texttt{fminbnd} along with the SDP procedure used for the other bounds computed in this work.

The results of performing this optimization on the HK18 model are displayed in Figure \ref{fig:NmaxK}. Upper bounds with $k^2 = 1/2$ are maximal at onset, and remain nearly optimal until $\R \approx 100 \,\R_{c}$. The optimizer $k^*$ is larger than $\sqrt{1/2}$ whenever $\R > \R_{c}$, and increases gradually with $\R$, reaching $(k^*)^2 \approx 0.656$ when $\R \approx 95\, \R_c$. In this regime, the upper bound is saturated by the primary equilibria with $k = k^*$. The maximum value of $N$ along the primary equilibria can be estimated by numerically continuing the equilibrium branch over $k$ at fixed $\R$; doing so appears to yield the same values of $k^*$ as the upper bounds when $\R_c \leq \R \leq 95 \,\R_c$. For slightly larger $\R$, the upper bound $N^*_{k^*}$ is saturated by a branch of equilibria that bifurcates from the primary branch, and the optimizer immediately jumps to $(k^*)^2 \approx 2.84$. The overall behavior of $N^*_{k^*}$ is similar for other models in the HK hierarchy with $M_i \leq 28$, but computation time increases sharply for larger models in the hierarchy.

Upper bounds with degree 4 $V$ were also constructed for various selected values of $k^2$ for models in the hierarchy up to HK28. Bounds at fixed $k$ are compared with the bounds at maximal $k$ in Figure \ref{fig:NmaxK}. Changing the value of $k$ primarily affects the $\R$ value where equilibria first bifurcate from the zero state, and slightly alters the shape of the upper bound curve when plotted against $\R$. Among the values of $k$ computed, the largest value of $U^*_d$ is attained with $k^2 = 1/2$ for $\R \lesssim 100 \,\R_{c}$, and with $k^2 = 2$ for larger $\R$. This agrees with the computed optimal values of $k^*$ found in the preceding discussion.

\begin{figure}[!tp]
\centering
\hspace{1cm}\includegraphics[scale=.5]{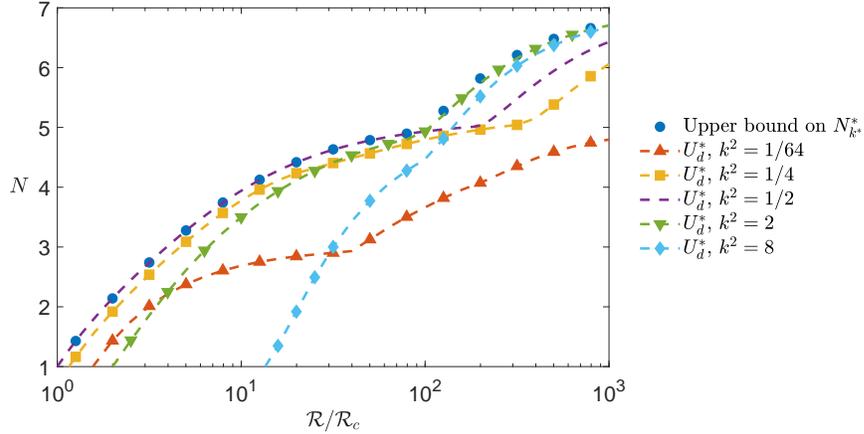}
\caption{Sum-of-squares upper bounds on the HK18 model maximized over $k$ compared with the upper bounds at selected fixed values of $k$. Bounds are computed with degree 4 $V$ and $\sigma = 10$.}
\label{fig:NmaxK}
\end{figure}


\section{Conclusions}
\label{sec:Con}

The models in the HK hierarchy are reduced-order models (ROMs) of Rayleigh--B\'enard convection that obey energy, temperature, and vorticity balance laws derived from the Boussinesq equations. Although previous works focused on conservation laws that hold in the dissipationless limit, we find that the same model construction criteria also satisfy analogous balance laws derived in the presence of dissipation. These models are expressed in a form amenable to procedural construction of the ODEs, and the various models analyzed in this work are constructed in this way. Each model in our hierarchy of distinguished models includes at least one shear mode of the form $\psi_{0n}$. Solutions where the shear modes are active are reminiscent of the zonal flows observed in various studies of Rayleigh--B\'enard convection in periodic domains. 

Various models in the HK hierarchy are analyzed by examining their bifurcation structure and computing upper bounds on the time-averaged heat transport using sum-of-squares optimization. We observe that $N$ is most often maximized by steady states, especially when the Rayleigh number is sufficiently small to allow quantitative comparison to the PDE. The primary branch of equilibria---the steady states that emerge as the first instability of the static state when $\R$ is raised---maximize heat transport for a range of $\R$ at the onset of convection. The heat transport of the primary states is enhanced when additional modes are included in the reduced-order model that pair with the nonzero variables of the primary branch. The first such enhancement occurs in the HK14 model, and additional jumps in the total heat transport typically occur when the HK hierarchy begins a ``shell" in the hierarchy that has even total wavenumber $(m+n)$. Other distinguished modal hierarchies could potentially be constructed that take advantage of this phenomenon to achieve greater heat transport with a similar number of modes.

We also observe that the states providing optimal heat transport do not include shear modes except when the Rayleigh number is well beyond the point where the reduced models closely approximate the heat transport of the full Boussinesq equations. In fact, we observe that the shear modes are identically zero along all equilibria that bifurcate from the zero state (the $L_{mn}$ equilibria and higher-dimensional analogues). Steady states exhibiting shear flow---analogues of the $TC$ equilibria studied in the HK8 model---were only observed to emerge as bifurcations from equilibria other than the zero state. The study of reduced-order models without shear is another possible direction for future research. Such models arise naturally if one considers a fully enclosed fluid domain, rather than imposing horizontal periodicity. Another possible direction of future work is to consider other types of boundary conditions, such as no-slip or fixed-flux conditions along the walls. 

\appendix

\section{Conservation properties} \label{sec:ConsLaws}

In this appendix, we derive the criteria that were used to construct models in the HK hierarchy from conservation laws of the Boussinesq equations~\eqref{eq:BE_psi}--\eqref{eq:BE_theta}. First we examine the restrictions on mode selection imposed by the conservation of energy, temperature and vorticity in the dissipationless limit $\nu, \kappa \to 0$ in \ref{sec:ConsLawsDLess}. Conservation laws in the ideal limit were imposed in various studies~\cite{Treve1982,Thiffeault1995,Hermiz1995,Gluhovsky2002}, and were applied to construct the HK8 model~\cite{Howard86,Goluskin2013,Olson2020}. In \ref{sec:DissipLaw} we show that these restrictions do not change if one considers the analogous integral balance laws derived from the full Boussinesq equations.

\subsection{Conservation laws in the dissipationless limit} \label{sec:ConsLawsDLess}

The dimensionless form of the Boussinesq equations~\eqref{eq:BE_psi}--\eqref{eq:BE_theta} is not amenable to taking the dissipationless limit, since some variable scalings depend on $\nu$ and $\kappa$. Instead, we nondimensionalize the equations by scaling length by $d$, time by $\sqrt{d/g\alpha \Delta T},$ the stream function by $\sqrt{g \alpha \Delta T d^3}$ and temperature by $\Delta T$ as in \cite{Thiffeault1995}. As $\nu, \kappa \to 0$, the governing equations become
\begin{align}
	\partial_t \nabla^2 \psi - \{\psi, \nabla^2 \psi \} &= \partial_x \theta, \label{eq:BE_diss_p}\\
    \partial_t \theta - \{\psi, \theta\} &= \partial_x \psi. \label{eq:BE_diss_t}
\end{align}
The equations \eqref{eq:BE_diss_p}--\eqref{eq:BE_diss_t} admit a number of conserved quantities \cite{Thiffeault1995,Gluhovsky2002}. One such quantity is the energy, $E = K + U$, where $K$ and $U$ satisfy
\begin{align}
K &= \frac{1}{2} \blangle \left\vert\nabla \psi \right\vert^2 \brangle, \\
U &= \blangle z \theta \brangle,
\end{align}
with the spatial average $\langle \cdot \rangle$ defined by \eqref{eq:spaceAvg}. To verify that this is conserved by the dissipationless Boussinesq equations, multiply \eqref{eq:BE_diss_p} by $\psi$ and average over the domain, imposing boundary conditions and integrating by parts when necessary, to obtain
\begin{equation}
\frac{1}{2} \partial_t \blangle \left\vert \nabla \psi \right\vert^2 \brangle = - \blangle \psi \partial_x \theta \brangle.
\end{equation}
Likewise, multiplying \eqref{eq:BE_diss_t} by $z$ and taking the volume average produces
\begin{align}
\partial_t \blangle z \theta \brangle  = \blangle \psi \partial_x \theta \brangle.
\end{align}
Adding these two expressions provides the desired result $\partial_t E = 0.$ 

Let $S_\psi$ and $S_\theta$ be the sets of Fourier mode pairs in a chosen truncated model, as in \S\ref{sec:constructModel}. Thiffeault and Horton~\cite{Thiffeault1995,Thiffeault1996} found that Galerkin-truncated models of Rayleigh's PDE conserve energy in the dissipationless limit if and only if the modes are chosen such that
\begin{crit}[Energy balance] \label{rule:Econs1}
If $(m,n) \in S_\psi \cap S_\theta,$ then $(0,2n) \in S_\theta$.
\end{crit}

The Lorenz equations have $S_\psi \cap S_\theta = \{(1,1)\}$ and $(0,2) \in S_\theta$, satisfying Criterion \ref{rule:Econs1}. On the other hand, the model of Howard and Krishnamurti includes $(1,2) \in S_\psi \cap S_\theta$ but is missing $(0,4) \in S_\theta$; adding $\theta_{04}$ restores the energy balance \cite{Thiffeault1995}. There are a few important consequences of selecting reduced models that satisfy the energy criterion. All trajectories of such models remain bounded, even in the presence of dissipation~\cite{Thiffeault1995}. This is significant because unbounded trajectories have been observed for certain ODE models \cite{Howard86}, marking a significant divergence from the physics of Rayleigh--B\'enard convection. Also, the two definitions of $N$, \eqref{eq:N_gen} and \eqref{eq:N2_gen}, are equivalent along all solutions of models satisfying the energy criterion \cite{Thiffeault1995}. Finally, models satisfying criterion \ref{rule:Econs1} conserve total temperature in the dissipationless limit, that is,
\begin{equation}
\partial_t \blangle \theta \brangle  = 0,
\end{equation} 
derived by taking the volume average of \eqref{eq:BE_diss_t}.

The Boussinesq equations also satisfy conservation of the integral of vorticity \cite{Gluhovsky2002} in the dissipationless limit:
\begin{equation}
\partial_t \Blangle \nabla^2 \psi \Brangle = 0,
\end{equation} 
determined by taking the spatial average of \eqref{eq:BE_diss_p}. Next we identify the criterion for reduced models to satisfy vorticity conservation in the dissipationless limit. Let the function $f_\psi(x) = \sin(x)$ when $m+n$ is even or $\cos(x)$ when $m+n$ is odd according to the convention established in \S\ref{sec:simplifyModel}. Projecting the integral of vorticity onto $S_\psi$ yields
\begin{align} \label{eq:VortCons1}
\blangle \nabla^2 \psi \brangle &= \left \langle \, -\sum_{\mathclap{(m,n) \in S_\psi}} \rho_{mn} \psi_{mn} \, f_\psi(mkx) \sin(nz) \right \rangle \\
&= -\sum_{\mathclap{(m,n) \in S_\psi}} \rho_{mn} \psi_{mn} \, \Blangle f_\psi(mkx) \sin(nz) \Brangle.
\end{align}
The volume average vanishes for each term with $m \neq 0$ and for all even $n$. For these terms, $\rho_{0n} = n^2$, so after integrating, the volume-averaged vorticity reduces to
\begin{align}
\blangle \nabla^2 \psi \brangle &= -\tfrac{1}{\pi} \sum_{\mathclap{(0,n) \in S_\psi}} \tfrac{1 - \cos(\pi n)}{n} n^2 \,\psi_{0n}\\
&= -\tfrac{1}{\pi} \sum_{\mathclap{\substack{(0,n) \in S_\psi \\ n \, \text{odd}}}} 2n \, \psi_{0n}. \label{eq:vortDiss}
\end{align}
Expressions for $\dot{\psi}_{0n}$ can be computed by adapting the general ROM equations \eqref{eq:psiTrunc} for the dissipationless scaling with $m=0$, resulting in:
\begin{equation}
\dot{\psi}_{0n} = \frac{k}{4n^2} \sum_{P_\psi[(0,n)]} (-1)^{p+s} p (s + B_{qns} q)(q^2 - s^2) \psi_{pq} \psi_{ps},
\end{equation}
where as in \S\ref{sec:constructModel}, the set $P_\psi[(m,n)]$ \eqref{eq:psiPairs} contains the wavenumber pairs of modes for each term contributing to the quadratic part of the right-hand side of the $\psi_{mn}$ ODE. The set $P_\psi[(0,n)]$ takes the form $\{(p,q),(p,s)\}$, where $p \neq 0$ and either $n = q + s$ or $n = \lvert q - s \rvert$, and correspond to terms proportional to $\psi_{pq} \, \psi_{ps}$. Note that the $\psi$ modes in the nonlinear term have the same horizontal mode since $m = 0$, and that $q > s$ due to the ordering placed on $P_\psi$ to avoid duplicate terms.

Taking the time derivative of \eqref{eq:vortDiss}, we obtain
\begin{align} 
\partial_t \langle \nabla^2 \psi \rangle &= -\tfrac{1}{\pi} \sum_{\mathclap{(0,n) \in S_\psi}} 2n \, \dot{\psi}_{0 n} \\
&=  -\tfrac{1}{\pi} \sum_{\mathclap{(0,n) \in S_\psi}} \hspace{2pt} \tfrac{k}{2n} \hspace{-2pt} \sum_{P_\psi[(0,n)]} (-1)^{p+s} p \left(s + B_{qns} \, q \right) (q^2 - s^2) \psi_{pq} \psi_{ps}. \label{eq:VortCons2}
\end{align}
Now, consider terms of the form $\psi_{pq} \, \psi_{ps}$ for fixed $p, q, s$. Such terms can appear at most two times in the sum \eqref{eq:VortCons2}, and this occurs in the terms generated from the $\psi_{0, q+s}$ and $\psi_{0, |q - s|}$ equations. Combining these terms results in
\begin{align}
-\frac{k}{2\pi} (-1)^{p+s}\, p \,(q^2 - s^2)\,\left[\tfrac{1}{(q + s)} (s + q) + \tfrac{1}{ q - s}( s - q )\right] \psi_{pq} \psi_{ps} = 0,
\end{align}
so long as all relevant modes are included in the truncated model. If one of $\psi_{0,q+s}$ or $\psi_{0,|q-s|}$ are included, but not the other, at least one term in \eqref{eq:VortCons2} remains.  Conservation of vorticity therefore imposes the following criterion on mode selection:
\begin{crit}[Vorticity balance] \label{rule:Vcons1}
If $(p,q) \in S_\psi$ and $(p,s) \in S_\psi$, then $(0,\lvert q-s\rvert ) \in S_\psi$ if and only if $(0,q+s) \in S_\psi$.
\end{crit}
The six-ODE and seven-ODE models discussed above include $(1,1), (1,2) \in S_\psi$ as well as the shear mode $(0,1) \in S_\psi$. Hence vorticity conservation is enforced by adding $(0,3) \in S_\psi$ \cite{Hermiz1995}.

\subsection{Integral balances in the presence of dissipation} \label{sec:DissipLaw}
For each of the conservation laws of the dissipationless Boussinesq equations, there exists an analogous integral balance derived from the full PDE \eqref{eq:BE_psi}--\eqref{eq:BE_theta}. The resulting energy, temperature and vorticity balance laws are
\begin{align}
\partial_t \left[ \tfrac{1}{2} \left\langle \lvert\nabla \psi \rvert^2 \right\rangle + \sigma \R \langle z \theta \rangle \right] &=  \, \sigma \R \left\langle z \nabla^2 \theta \right\rangle - \sigma \left\langle (\nabla^2 \psi )^2 \right\rangle, \label{eq:EConsD} \\[6pt]
\partial_t \langle \theta \rangle &=  \, \left\langle \nabla^2 \theta \right\rangle, \label{eq:TConsD}\\[6pt]
\partial_t \left\langle \nabla^2 \psi \right\rangle &= \, \sigma \left\langle \nabla^4 \psi \right\rangle .\label{eq:VConsD}
\end{align}
Here we show that if a reduced-order model obeys conservation of energy, temperature and vorticity in the dissipationless limit, it also satisfies the corresponding balance equations for the PDE with dissipation.

If the truncated Fourier expansions for $\psi$ and $\theta$ are substituted into \eqref{eq:EConsD}, orthogonality reduces the left-hand side of the energy balance to
\begin{align}
\partial_t E &= \sum_{\mathclap{\substack{(m,n) \in S_\psi\\m\neq 0}}} \tfrac{1}{4} \rho_{mn} \psi_{mn} \dot{\psi}_{mn} + \sum_{\mathclap{(0,n) \in S_\psi}} \tfrac{1}{2} n^2 \psi_{0n} \dot{\psi}_{0n}-\sigma \R \sum_{\mathclap{(0,n) \in S_\theta}} \tfrac{1}{2n} \dot{\theta}_{0,2n}.
\end{align}
The proof in \cite{Thiffeault1995} shows that the above expression vanishes in the absence of dissipation provided the modes are selected as specified in \ref{sec:ConsLawsDLess}. Here, similar cancellation occurs, leaving only the terms resulting from the dissipative terms of \eqref{eq:BE_psi}--\eqref{eq:BE_theta}:
\begin{align}
\partial_t E = - \sum_{\mathclap{\substack{(m,n) \in S_\psi\\m\neq 0}}} \tfrac{1}{4} \rho_{mn}^2 \sigma \psi_{mn}^2 - \sum_{\mathclap{(0,n) \in S_\psi}} \tfrac{1}{2} n^4 \sigma \psi_{0n}^2  + \sigma \R \sum_{\mathclap{(0,n) \in S_\theta}} \tfrac{n}{2} \theta_{0,2n}.
\end{align}
This can also be proved in a similar manner as in the proof of the vorticity conservation law in \ref{sec:ConsLawsDLess}. The projection of the right-hand side of \eqref{eq:EConsD} onto any Fourier-truncated $\psi$ and $\theta$ is identical to the above expression, so the general energy balance holds under the exact same conditions as its dissipationless version. The truncated version of the integral balance for temperature \eqref{eq:TConsD} is trivial to prove since the average temperature vanishes for all Fourier modes in the expansion for $\theta$. The vorticity balance \eqref{eq:VConsD} is also easy to show. First, following the proof in \ref{sec:ConsLawsDLess}, the spatially averaged vorticity in the presence of dissipation is:
\begin{equation}
\partial_t \blangle \nabla^2 \psi \brangle = \frac{1}{\pi} \sum_{(0,n) \in S_\psi} 2n (\sigma n^2 \psi_{0n}).
\end{equation} 
Similarly, the right-hand side of \eqref{eq:VConsD} simplifies to 
\begin{equation}
\sigma \blangle \nabla^4 \psi \brangle = \sigma \frac{1}{\pi} \sum_{(0,n) \in S_\psi} n^4 \frac{2}{n} \psi_{0n}.
\end{equation}
Therefore, the vorticity balance law holds under the same criterion as its dissipationless version.

In this appendix, we established criteria on the mode selection for reduced-order models of Rayleigh--B\'enard convection such that the resulting models obey truncated versions of energy, temperature and vorticity balance laws derived from the PDE. We call low-order models that obey each of the balance laws above \emph{distinguished models}. Each model in the HK hierarchy that is derived in \S\ref{sec:HK} is a distinguished model in this sense.

\section{Numerical procedure}
\label{sec:NumProc}


In this appendix, we discuss the numerical procedure used to compute upper bounds in \S\ref{sec:UB} according to the sum-of-squares optimization process. The number of terms in the general ansatz for the auxiliary function $V \in \mathbb{P}_{n,d}$ grows rapidly in both the dimension $n$ of the ODE and the maximum degree $d$ of the monomials in the ansatz. Increasing either $n$ or $d$ results in significant increases in computational cost and poor numerical conditioning in all but the smallest SOS problems. These issues can be remedied in part by taking advantage of the structure of the ODEs to reduce the number of monomials in the auxiliary function ansatz. Numerical conditioning can be further improved by scaling the phase space variables in the governing ODE system. Monomial reduction for SDP computations in this work was automated using Python's symbolic manipulation package \texttt{sympy}. The Python scripts are posted on GitHub\footnote{GitHub repository: https://github.com/PeriodicROM/ReduceMonomsRBC}. In this section, we detail how monomial reduction and scaling were accomplished in our numerical procedure.

\begin{table}[ht]
\centering
\ra{1.1}
\caption{Number of monomials up to degree 4 compared with the number in the reduced form of the ansatz for the auxiliary function $V$ for selected models in the HK hierarchy. Before reduction, there are $\binom{M_i+4}{4}$ monomials in the degree 4 ansatz, where $M_i$ is the dimension of the model. The time required to solve the SDP is reported for both the reduced and unreduced problems (the unreduced problem was not solved for $M_i > 28$ due to memory constraints). As a rule of thumb, the memory and time requirements scale roughly as $O(n^3)$ when the corresponding Gram matrix is of dimension $n$. Computation time for the HK26 model was comparatively slow because it has only one sign symmetry, while all others displayed here admit 3 symmetries.} \label{tab:MonomsHK} 
\begin{tabular}{cccccc}
\toprule
 & \multicolumn{2}{c}{Unreduced} && \multicolumn{2}{c}{Reduced}\\
 \cmidrule{2-3} \cmidrule{5-6}
Model & Monomials & Time (s) && Monomials & Time (s)\\
\midrule
HK8 & 495 & 2 && 89 & 0.5 \\
HK10 & 1001 & 5 && 159 & 0.7  \\
HK14 & 3060 & 60 && 382 & 3 \\ 
HK16 & 4845 & 230 && 448 & 7\\ 
HK18 & 7315 & 450 && 575 & 9 \\ 
HK22 & 14950 & - && 978 & 50 \\ 
HK24 & 20475 & - && 1190 & 100 \\
HK26 & 27405 & - && 1434 & 815 \\
HK28 & 35960 & - && 1698 & 250\\
\bottomrule
\end{tabular}
\end{table}

Symmetry conditions can be used to improve numerical performance of SDP computations. Suppose that both $\Phi$ and the ODE are invariant under a symmetry given by the linear transformation $\Lambda$, so that $\Phi(\Lambda \bx) = \Phi(\bx)$ and $\mathbf{f}(\Lambda \bx) = \Lambda \mathbf{f}(\bx).$ Then any bound proved using the auxiliary function method can be proved with symmetric $V$, so that $V(\Lambda \bx) = V(\bx)$ \cite{Goluskin2019,Lakshmi2020}. Symmetry reductions are convenient to implement for sign-symmetries of the variables in $\bx$, where $\Lambda$ is a diagonal matrix such that each diagonal entry is $\pm 1$. Given $\Phi$ and $\mathbf{f}$, let monomials be represented in vector form by multi-indices $\boldsymbol{\alpha} \in \mathbb{Z}^n$, where
\begin{equation}
\bx^{\boldsymbol{\alpha}} = x_1^{\alpha_1} x_2^{\alpha_2} \cdots x_n^{\alpha_n}.
\end{equation}
Next, let a symmetry $\Lambda$ of the form described above be represented by $\mathbf{s} \in \mathbb{Z}_2^n$, with $s_i = 1$ if $x_i$ is reflected under $\Lambda$, and $s_i = 0$ otherwise. For example, the vector $(1,1,0)$ corresponds to the symmetry $(x,y,z) \mapsto (-x,-y,z)$. Within this framework, a monomial is invariant under the symmetry $\Lambda$ if and only if $\mathbf{s} \cdot \boldsymbol{\alpha} \equiv 0$ (mod 2). This provides a computationally efficient way to determine if a candidate monomial in the general ansatz for $V$ is symmetric. The set of symmetries for a given problem can be determined in much the same way. First, construct the matrix $\mathbf{A}$ whose rows are the multi-indices of each term in the polynomials $x_i f_i(\bx)$ and $\Phi(\bx)$. Any symmetry must then satisfy $\mathbf{A} \mathbf{s} \equiv 0$ (mod 2). Keeping only the symmetric terms in the ansatz for $V$ reduces time and memory constraints since fewer coefficients must be determined. It also yields Gram matrices that are readily block diagonalizable, further improving numerical performance of the SDP algorithm.

The sum-of-squares constraint on the polynomial $S$ \eqref{eq:S} implies that its highest-degree terms must be of even degree. This fact can be exploited to simplify the ansatz for $V$. For all models in this paper, the function $\mathbf{f}$ has degree 2, so in general the polynomial $S = U - \mathbf{f} \cdot \nabla V$ has degree $d+1$ for $V$ of degree $d$. When $d$ is even, $S$ admits a sum-of-squares factorization only if the highest-degree terms of $\mathbf{f} \cdot \nabla V$ cancel. Any terms in $V$ that are incompatible with this highest-degree cancellation condition can be discarded without affecting the upper bound. The coefficients of the discarded terms must be zero in any $V$ satisfying the SOS constraint. We apply the highest-degree cancellation condition in all SDP computations in this work, resulting in a reduced monomial basis for $V$. Doing so improves numerical conditioning and computational complexity of the SDP algorithm.

The numerical conditioning of SDP optimization can be significantly improved by scaling the state variables of the ODE. For all SDP computations performed in this work, the state variables were scaled so that all relevant trajectories are contained roughly within the region $[-1,1]^n$. This heuristic has been employed for SDP computations in other works, and in such cases doing so improved the numerical conditioning of the SDP~\cite{Goluskin2018,Henrion2014}.

When the ODEs are expressed in the form \eqref{eq:psiTrunc}--\eqref{eq:thetaTrunc}, the variable scalings required by the above criteria often change significantly as the Rayleigh number increases. For computational purposes this can be remedied either by constructing scaling factors that change with $\R$, or by making a change of variables in the governing equations. The latter is the approach taken in this work, and is accomplished by letting $\psi'$ and $t'$ be given by
\begin{equation} \label{eq:scaledVars}
\psi' := \R^{-1/2} \psi, \qquad t' := \R^{1/2} t.
\end{equation}
This change of variables is motivated by the fact that $\psi \sim O(\R^{1/2})$ in the original expressions for the $L_{mn}$ equilibria \eqref{eq:Lmn_simple}. Under the transformation \eqref{eq:scaledVars}, the expressions governing the truncated models take the form
\begin{align}
    \dot{\psi'}_{mn} &= -\sigma \R^{-1/2} \rho_{mn} \psi'_{mn} + (-1)^{m+n} \sigma \frac{mk}{\rho_{mn}} \theta_{mn} + \frac{k}{\rho_{mn}} Q_{mn}^\psi, \label{eq:psiTruncNew} \\
     \dot{\theta}_{mn} &= -\rho_{mn} \R^{-1/2} \theta_{mn} + (-1)^{m+n} (mk) \psi'_{mn} + k Q_{mn}^\theta, \label{eq:thetaTruncNew}.
\end{align}
Many of the expressions defined in \S\ref{sec:constructModel} are unaffected by this transformation, including the quadratic terms $Q_{mn}^\psi$ \eqref{eq:psiTruncQuad} and $Q_{mn}^\theta$ \eqref{eq:thetaTruncQuad}, and both versions of $N$ defined in \eqref{eq:N_gen}--\eqref{eq:N2_gen}.

The change of variables \eqref{eq:scaledVars} appears to provide the proper scaling to yield dynamics roughly within $[-1,1]^n$. In some cases, solutions deviate slightly from this region, leading to poor numerical conditioning. This can be rectified by performing the uniform rescaling $\mathbf{x} \mapsto \alpha \mathbf{x}$ for some empirically determined constant $\alpha$; for computations in this work we set $\alpha = 2$ to ensure the proper scaling.

After completing the above pre-processing steps, the toolbox YALMIP \cite{Lofberg2004, Lofberg2009} (version 20190425) is used to formulate the optimization problem in the form of \eqref{eq: bound3} and pass the problem to the solver, MOSEK \cite{mosek} (version 9.0.98).  The upper bounds presented in this work were computed using a 3.0 GHz Xeon processor.

\section{Truncated model examples: Lorenz and HK8} \label{sec:examples}

The model construction process outlined in \S\ref{sec:Galerkin} can be used to construct a variety of truncated models of Rayleigh--B\'enard convection. In this section, the process of constructing such models is outlined for two particular examples: the Lorenz equations \cite{Lorenz63} and the HK8 model \cite{Gluhovsky2002}. 

Recall that the model construction requires selecting a number of modes for $\psi$ and $\theta$ that are collected in the sets $S_\psi$ and $S_\theta$, respectively. These sets are then used to build compatible triples of modes that will appear in the quadratic terms; these triples take the form 
\begin{align}
P_\psi[(m,n)] &= \{ ((p,q), (r,s)) \in S_\psi \times S_\psi: m = |p \pm r|, n = |q \pm s|, (p,q) > (r,s) \} \\
P_\theta[(m,n)] &= \{ ((p,q), (r,s)) \in S_\psi \times S_\theta: m = |p \pm r|, n = |q \pm s| \}, 
\end{align}
where $(p,q) > (r,s)$ refers to the lexicographical ordering. The general equations described in \S\ref{sec:constructModel} are:
\begin{align} 
    \dot{\psi}_{mn} &= -\sigma \rho_{mn} \psi_{mn} + (-1)^{m+n} (\sigma \R) \frac{mk}{\rho_{mn}} \theta_{mn} + Q_{mn}^{\psi}, \label{eq:psiTrunc2} \\
     \dot{\theta}_{mn} &= -\rho_{mn} \theta_{mn} + (-1)^{m+n} (mk) \psi_{mn} + Q_{mn}^{\theta}, \label{eq:thetaTrunc2}
 \end{align}
where $\rho_{\alpha}$ are the eigenvalues and $Q^psi$, $Q^\theta$ are given by
\begin{align}
    Q_{mn}^\psi &= \frac{k}{\rho_{mn}} \sum_{P_\psi[(m,n)]} \frac{\mu_1}{d} \big[B_{pmr} B_{snq} (p s) - B_{qns} (q r) \big] (\rho_{p q} - \rho_{r s}) \psi_{p q} \psi_{r s},  \label{eq:psiTruncQuad2} \\
    Q_{mn}^\theta &= k \sum_{P_\theta[(m,n)]} \frac{\mu_2}{d} \big[B_{pmr} B_{snq} (p s) - \mu_3 B_{qns} B_{rpm}(q r) \big] \psi_{pq} \theta_{rs}, \label{eq:thetaTruncQuad2}
\end{align}
Here $B,$ $\mu_1,$ $\mu_2,$ $\mu_3$ and $d$ are defined by
\begin{align}
    B_{ijk} &= \begin{cases}
        -1, & i = j + k, \\
        1, & \mbox{ else},
        \end{cases} \displaybreak[1] \label{eq:Bmatrix2}\\
    \mu_1 &= \begin{cases}
        B_{pmr}, & (m+n) \ \mbox{ even}, (r+s) \  \mbox{ odd}, \\
        -B_{pmr},  & (m+n) \ \mbox{ odd},\; (r+s) \  \mbox{ odd}, \\
        -1, & {\rm else},
        \end{cases} \displaybreak[1] \label{eq:mu1}\\
    \mu_2 &= \begin{cases}
    	\mu_3 \,B_{rpm}, & (m+n) \ \mbox{ even}, (r+s) \ \mbox{ odd}, \\
    	-B_{rpm} B_{pmr}, & (m+n) \ \mbox{ odd}, \;(r+s) \ \mbox{ even}, \\
    	B_{pmr}, & (m+n) \ \mbox{ even}, (r+s) \ \mbox{ even},\\
    	1 & \mbox{ else},
    	\end{cases}\displaybreak[1] \label{eq:mu2}\\
    \mu_3 &= \begin{cases}
        -1, & m=0, \\
        1, & \mbox{else},
        \end{cases} \displaybreak[1] \label{eq:mu3} \\
    d &= \begin{cases}
        2, & p = 0 \ \mbox{or} \  r = 0, \\
        4, & \mbox{else}.
        \end{cases} \label{eq:d}
\end{align}

\subsection{Lorenz equations}

The Lorenz equations can be constructed from the above equations by selecting the modes $\psi_{11}$, $\theta_{11}$, and $\theta_{02}$. The corresponding eigenvalues are $\rho_{11} = (k^2+1)$, and $\rho_{02} = 4$. The compatible triples for each mode are:
\begin{align}
P_\psi[(1,1)] &= \{\}, \\
P_\theta[(1,1)] &= \{\big( (1,1), (0,2) \big) \}, \\
P_\theta[(0,2)] &= \{\big( (1,1), (1,1) \big) \}.
\end{align} 
Therefore, the $\psi_{11}$ equation has no quadratic terms, and the other two equations have one quadratic term each. These quadratic terms are determined by computing the values of the constants in \eqref{eq:Bmatrix2}--\eqref{eq:d}. For the $Q_{11}^\theta$ term, these constants are:
\begin{alignat}{3}
B_{pmr} &= -1, \qquad \qquad & B_{snq} &= -1, \qquad \qquad & B_{qns} &= 1, \\
B_{rpm} &= 1, & \mu_2 &= -1, &\mu_3 &= 1, \\
&& d &= 2.
\end{alignat}
As a result,
\begin{equation}
Q_{11}^\theta = -\tfrac{1}{2} \left[ (-1)(-1)(2) - (1)(1)(1)(0) \right] = - \psi_{11} \, \theta_{02}.
\end{equation}
Following the same procedure, $Q_{02}^\theta = \tfrac{1}{2} \psi_{11} \, \theta_{11}$. Inserting these quadratic terms into \eqref{eq:psiTrunc2}--\eqref{eq:thetaTrunc2} yields the ODEs
\begin{align}
\dot{\psi}_{11} &= -\sigma (k^2+1) \psi_{11} + \sigma \R \frac{k}{k^2+1} \theta_{11}, \\
\dot{\theta}_{11} &= -(k^2+1) \theta_{11} + k \psi_{11} - k \psi_{11} \theta_{02}, \\
\dot{\theta}_{02} &= -4 \theta_{02} + \tfrac{k}{2} \psi_{11} \theta_{11}.
\end{align}
This system of equations can be transformed into the standard form of the Lorenz equations using the change of variables:
\begin{equation}
x = \frac{k}{\sqrt{2} \rho_{11}} \psi_{11}, \qquad y = \frac{k^2}{\sqrt{2} (\rho_{11})^3} \R \theta_{11}, \qquad z = \frac{k^2}{(\rho_{11})^3} \R \theta_{02}, \qquad \tau = (k^2 + 1) t.
\end{equation}
We then obtain
\begin{align}
\dot{x} &= \sigma (y - x), \\
\dot{y} &= -y + x (r - z), \\
\dot{z} &= -\beta z + xy,
\end{align}
where $r = \R k^2/(k^2+1)^3$ and $\beta = 4/(k^2+1)$. The Lorenz model is not included in the hierarchy given in \S\ref{sec:HK} since it does not have any shear modes. However, augmenting the Lorenz equations with the mode $\psi_{01}$ produces the first model in the hierarchy (HK4). The additional ODE is simply $\dot{\psi}_{01} = - \sigma \psi_{01}$, since the new mode does not form a compatible triple with any pair of modes in the Lorenz system. Thus, solutions to the HK4 model rapidly approach those of the Lorenz equations.

\subsection{HK8 Model}

The HK8 model is an extension of the Lorenz equations that was first considered in \cite{Gluhovsky2002} and was further analyzed in \cite{Goluskin2013, Souza2015, Olson2020}. The HK8 model equations can be reproduced by adding the modes $\psi_{01},\, \psi_{03}, \, \psi_{12},\, \theta_{04},$, and $\theta_{12}$ to the expansion that was used to construct the Lorenz equations. This results in more compatible triples:
\begin{equation}
\begin{aligned}
P_\psi[(0,1)] &= \{ ((1,1), (1,2)) \}, \\
P_\psi[(0,3)] &= \{ ((1,1), (1,2)) \}, \\
P_\psi[(1,1)] &= \{ ((0,1), (1,2)), ((0,3), (1,2)) \}, \\
P_\psi[(1,2)] &= \{ ((0,1), (1,1)), ((0,3), (1,1)) \}, \\
P_\theta[(0,2)] &= \{ ((1,1), (1,1)) \}, \\
P_\theta[(0,4)] &= \{ ((1,2), (1,2)) \}, \\
P_\theta[(1,1)] &= \{ ((1,1), (0,2)), ((0,1), (1,2)), ((0,3), (1,2)) \}, \\
P_\theta[(1,2)] &= \{ ((1,2), (0,4)), ((0,1), (1,1)), ((0,3), (1,1)) \}.
\end{aligned}
\end{equation}
Each member of these sets produces a quadratic term on the right-hand side of the corresponding ODE. The $\psi_{11}$ equation, for instance, now contains two quadratic terms---one proportional to $\psi_{01} \psi_{12}$ and the other proportional to $\psi_{03} \psi_{12}$. The coefficients of these terms are generated in the same manner as was demonstrated for the Lorenz equations. The result is the HK8 model: 
\begin{equation}
\label{eq:HK8}
\begin{aligned}
\dot{\psi}_{11} &= -\sigma (k^2+1) \psi_{11} + \sigma \R \tfrac{k}{k^2+1} \theta_{11} + \tfrac{k}{2} \tfrac{k^2+3}{k^2+1} \psi_{01} \psi_{12} - \tfrac{3k}{2} \tfrac{k^2-5}{k^2+1} \psi_{12} \psi_{03}, \\
\dot{\psi}_{01}  &= -\sigma \, \psi_{01} - \tfrac{3k}{4} \psi_{11} \psi_{12},\\
\dot{\psi}_{12} &= -\sigma (k^2+4) \psi_{12} - \sigma \R \tfrac{k}{k^2+4} \theta_{12} - \tfrac{1}{2} \tfrac{k^3}{k^2+4} \psi_{11} \psi_{01} + \tfrac{3k}{2} \tfrac{k^2-8}{k^2+4} \psi_{11} \psi_{03}, \\
\dot{\theta}_{11} &= -(k^2+1) \theta_{11} + k \psi_{11} - k \psi_{11} \theta _{02}  { - \tfrac{k}{2} \psi_{01} \theta_{12} }  {\ +  \tfrac{3k}{2} \theta_{12} \psi_{03}}, \\
\dot{\theta}_{02} &= -4 \, \theta_{02} + \tfrac{k}{2} \psi_{11} \theta_{11}, \\
\dot{\theta}_{12}  &= -(k^2+4) \theta_{12} - k \psi_{12} + \tfrac{k}{2} \psi_{01} \theta_{11} - \tfrac{3k}{2} \psi_{03} \theta_{11} + 2k \psi_{12} \theta_{04}, \\
\dot{\psi}_{03} &= -9\, \sigma \,  \psi_{03} + \tfrac{k}{4} \psi_{11} \psi_{12},\\
\dot{\theta}_{04} &= -16\,  \theta_{04} - k \psi_{12} \theta_{12}.
\end{aligned}
\end{equation}

\section{Proofs} \label{sec:proofs}

In the HK hierarchy, the zero state undergoes one or more pitchfork bifurcations due to linear instabilities, resulting in the \emph{primary equilibria} of the given HK model. We observed in \S\ref{sec:HKEquil} that shear modes, i.e. modes of the form $\psi_{0,n_0}$, are identically zero along primary equilibrium branches. We prove this below.

\begin{proposition}
Modes of the form $\psi_{0,n_0}$ are identically zero along all primary equilibria of any model in the HK hierarchy.
\end{proposition}

\begin{proof}
A linear perturbation analysis of the governing ODEs for models in the HK hierarchy \eqref{eq:psiTrunc}--\eqref{eq:thetaTrunc} reveals that bifurcations from the zero state occur in the $\psi_{mn}$--$\theta_{mn}$ subspace, for some modal pair $(m,n)$. Here we let $(m,n)$ be fixed and analyze the corresponding equilibrium branch. 

Additional variables can become nonzero along the $(m,n)$ branch when one or more nonzero terms appear on the right-hand side of the corresponding ODE. This occurs due to nonlinear pairing in the terms $Q_{m'n'}^\psi$ \eqref{eq:psiTruncQuad} or $Q_{m'n'}^\theta$ \eqref{eq:thetaTruncQuad}. All of the terms in $Q_{m'n'}^\psi$ are proportional to $\psi_{pq} \psi_{rs}$, where $m' = |p \pm r|$ and $n' = |q \pm s|$; likewise, the terms in $Q_{m'n'}^\theta$ are proportional to $\psi_{pq} \theta_{rs}$. For example, the $\theta_{0,2n}$ term will be nonzero since a term proportional to $\psi_{mn} \theta_{mn}$ appears in $Q_{0,2n}^\theta$, where under the notation above $m' = 0$, $n' = 2$, $p = r = m$, and $q = s = n$.

To prove the proposition, it suffices to show that this nonlinear pairing mechanism can never activate a mode of the form $\psi_{0, n_0}$. Here it is important to note that due to the horizontal phase condition introduced in \S\ref{sec:simplifyModel}, $n_0$ must be an odd integer in all shear modes. The key observation in the proof is that the terms that are activated in the nonlinear pairing mechanism are a subset of a vector space with integer scalars, spanned by $(m,n)$ and $(m,-n)$. This is because modes pair by the addition or subtraction of their wavenumbers, so all modes that are activated must be multiples of these building blocks. To see if the shear modes can be nonzero along the $(m,n)$ primary branch, we simply check if they lie in this vector space:
\begin{equation}
(0,n_0) = a (m,n) + b (m, -n).
\end{equation}
When $n_0$ is an odd integer, and $m, n$ are integers, one can easily check that the above system has no integer solutions. Therefore, shear modes are never activated by nonlinear pairing in the primary equilibria.
\end{proof}

\bibliography{HK}
\bibliographystyle{abbrv}

\end{document}